\documentclass[preprint]{imsart}
\RequirePackage{natbib}
    \pdfoutput=1
	\usepackage[english]{babel}
	\usepackage{amsmath}
	\usepackage{amssymb}
        
 	\usepackage{amsthm}
	\usepackage{bbm}
	\usepackage{xypic}
	\usepackage[all]{xy}
	\usepackage[pdftex]{graphicx}
	\usepackage[abs]{overpic}
	\usepackage{algorithm}
	\usepackage{algorithmic}
        \usepackage{enumerate}
	\usepackage{upgreek}
	\usepackage{textcomp}
	\usepackage{mathrsfs}
	\usepackage{leftidx}
	\usepackage{bigints}
	\usepackage{abbrevs}
	\usepackage{placeins}
	\usepackage[colorlinks=false]{hyperref}
        \usepackage{xcolor}
        \usepackage{xspace}
        	\usepackage{subfigure}

\usepackage{rotating}
\usepackage{tikz}

\startlocaldefs
	\newtheoremstyle{indent}{\topsep}{\topsep}{\addtolength{\leftskip}{1.5em}\itshape}{-1.5em}{\bfseries}{.\,}{ }{}
	\theoremstyle{indent}
	\newtheorem{corollary}{Corollary}
	
	\newtheorem{theorem}{Theorem}
	\newtheorem{cond}{Condition}

        \newtheorem{proposition}{Proposition}

    \newcommand{\Kappa}{K}
    \newcommand{\mvbar}{\middle\vert}
	\newcommand{\ud}{\,\mathrm{d}}

	\newcommand{\U}{\text{U}}
	
	\newcommand{\N}{\text{N}}
	
	\newcommand{\Exp}{\text{Exp}}

    \newcommand{\rem}[1]{{#1}^{\text{rem}}}
    
    \def\bx{{\mathbf{x}}}

    \newcommand{\algref}[1]{\hyperref[#1]{Algorithm \ref*{#1}}}
    \newcommand{\stepref}[1]{\hyperref[#1]{Step \ref*{#1}}}
    \newcommand{\algstref}[2]{\hyperref[#2]{Algorithm \ref*{#1} Step \ref*{#2}}}
    \newcommand{\figref}[1]{\hyperref[#1]{Figure \ref*{#1}}}
    \newcommand{\subfigref}[3]{\hyperref[#1]{Figure \ref*{#2}(#3)}}
    \newcommand{\tabref}[1]{\hyperref[#1]{Table \ref*{#1}}}
    \newcommand{\apxref}[1]{\hyperref[#1]{Appendix \ref*{#1}}}
    \newcommand{\apxrefpl}[2]{Appendices \hyperref[#1]{\ref*{#1}} and \hyperref[#2]{\ref*{#2}}}
    \newcommand{\secref}[1]{\hyperref[#1]{Section \ref*{#1}}}
    \newcommand{\prinref}[1]{\hyperref[#1]{Principle \ref*{#1}}}
    \newcommand{\conref}[1]{\hyperref[#1]{Condition \ref*{#1}}}
    \newcommand{\resref}[1]{\hyperref[#1]{Result \ref*{#1}}}
    \newcommand{\defnref}[1]{\hyperref[#1]{Definition \ref*{#1}}}
    \newcommand{\thmref}[1]{\hyperref[#1]{Theorem \ref*{#1}}}
    \newcommand{\lemref}[1]{\hyperref[#1]{Lemma \ref*{#1}}}
        \newcommand{\propref}[1]{\hyperref[#1]{Proposition \ref*{#1}}}
    \newcommand{\corrolref}[1]{\hyperref[#1]{Corollary \ref*{#1}}}
    \newcommand{\remref}[1]{\hyperref[#1]{Remark \ref*{#1}}}

\newcommand{\dtss}[1][]{\ensuremath{E}_{#1}\xspace}
\newcommand{\dtproc}[1][]{\ensuremath{\mathfrak{X}_{#1}}\xspace}
\newcommand{\dtaux}[1][]{\ensuremath{\mathfrak{Z}_{#1}}\xspace}

    \def\precon{\mathbf{\Lambda}}
\endlocaldefs

\def\bx{\mathbf{x}}
\def\bX{\mathbf{X}}
\def\by{\mathbf{y}}

\begin{document}
\begin{frontmatter}
\title{Quasi-stationary Monte Carlo and the ScaLE Algorithm}
\runtitle{Quasi-stationary Monte Carlo}

\begin{aug}
\author{\fnms{Murray} \snm{Pollock}\thanksref{e1}, \ead[label=e1,mark]{m.pollock@warwick.ac.uk}}
\author{\fnms{Paul} \snm{Fearnhead}\thanksref{e2}, \ead[label=e2,mark]{p.fearnhead@lancs.ac.uk}}\\
\author{\fnms{Adam M.} \snm{Johansen}\thanksref{e3} \ead[label=e3,mark]{a.m.johansen@warwick.ac.uk}}
\and
\author{\fnms{Gareth O.} \snm{Roberts}\thanksref{e4}
  \ead[label=e4,mark]{gareth.o.roberts@warwick.ac.uk}}

\printead{e1,e2}\\
\printead*{e3,e4}

\runauthor{M. Pollock et al.}
\affiliation{University of Warwick}
\end{aug}

\begin{abstract}
This paper introduces a class of Monte Carlo algorithms which are based upon the simulation of a Markov process whose quasi-stationary distribution coincides with a distribution of interest. This differs fundamentally from, say, current Markov chain Monte Carlo methods  which simulate a Markov chain whose stationary distribution is the target. We show how to approximate distributions of interest by carefully combining sequential Monte Carlo methods with methodology for the exact simulation of diffusions. The methodology introduced here is particularly promising in that it is applicable to the same class of problems as gradient based Markov chain Monte Carlo algorithms but entirely circumvents the need to conduct Metropolis-Hastings type accept/reject steps whilst retaining {\em exactness}: the paper gives theoretical guarantees ensuring the algorithm has the correct limiting target distribution. Furthermore, this methodology is highly amenable to big data problems. By employing a modification to existing  na\"{i}ve sub-sampling and control variate techniques it is possible to obtain an algorithm which is still exact but has {\em sub-linear} iterative cost as a function of data size.
\end{abstract}

\begin{keyword}
\kwd{Control variates} \kwd{Importance sampling} \kwd{Killed Brownian motion}\kwd{Langevin diffusion}  \kwd{Markov chain Monte Carlo} \kwd{Quasi-stationarity} \kwd{Sequential Monte Carlo} 
\end{keyword}

\end{frontmatter}


\TMResetAll
\section{Introduction} \label{s:intro}

Advances in methodology for the collection and storage of data have led to scientific challenges and opportunities in a wide array of disciplines. This is particularly the case in Statistics as the complexity of appropriate statistical models 
often increases with data size. Many current state-of-the-art statistical methodologies have algorithmic cost that scales poorly with increasing volumes of data. As noted by \cite{b:j13}, `many statistical procedures either have unknown runtimes or runtimes that render the procedure unusable on large-scale data' and this has resulted in a proliferation in the literature of methods `\ldots{}which may provide no statistical guarantees and which in fact may have poor or even disastrous statistical properties'.\\
\\
This is particularly keenly felt in computational and Bayesian statistics, in which the standard computational tools are Markov chain Monte Carlo (MCMC), Sequential Monte Carlo (SMC) and their many variants (see for example \cite{bk:mcsm}). MCMC methods are {\em exact} in the (weak) sense that they construct Markov chains which have the correct limiting distribution. Although MCMC methodology has had considerable success in being applied to a wide variety of substantive areas, they are not well-suited to this new era of `big data' as their computational cost will increase at least linearly with the number of data points. For example, each iteration of the Metropolis-Hastings algorithm requires evaluating the likelihood, the calculation of which, in general, scales linearly with the number of data points.
%
The motivation behind the work presented in this paper is on developing Monte Carlo methods that are exact, in the same sense as MCMC, but that have a have a computational cost per effective sample size that is sub-linear in the number of data points.\\
\\
To date, the success of methods that aim to adapt MCMC to reduce its algorithmic cost has been mixed and has invariably led to a compromise on exactness --- such methodologies generally construct a stochastic process with limiting distribution which is (at least hopefully) close to the desired target distribution. Broadly speaking these methods can be divided into three classes of approach: `Divide-and-conquer' methods; `Exact Sub-sampling' methods; and, `Approximate Sub-sampling' methods. Each of these approaches has its own strengths and weaknesses which will be briefly reviewed in the following paragraphs.\\
\\
Divide-and-conquer methods (for instance, \cite{uap:nwx14,arxiv:wd13,ijmsem:sbb16,icml:ssld14}) begin by splitting the data set into a large number of smaller data sets (which may or may not overlap). Inference is then conducted on these smaller data sets and resulting estimates are combined in some appropriate manner. A clear advantage of such an approach is that inference on each small data set can be conducted independently, and in parallel, and so if one had access to a large cluster of computing cores then the computational cost could be significantly reduced. 
The primary weakness of these methods is that the recombination of the separately conducted inferences is inexact. All current theory is asymptotic in the number of data points, $n$ \citep{uap:nwx14,b:lsd17}. For these asymptotic regimes the posterior will tend to a Gaussian distribution \cite{ams:j70}, and it is questionable whether divide-and-conquer methods offer an advantage over simple approaches such as a Laplace approximation to the posterior \cite{arxiv:bdh15}. Most results on convergence rates (e.g. \cite{aistats:sctd16}) have rates that are of order $\mathcal{O}(m^{-1/2})$, where $m$ is the number of data-points in each sub-set. As such they are no stronger than convergence rates for analysing just a single batch. One exception is in \cite{b:lsd17}, where convergence rates of $\mathcal{O}(n^{-1/2})$ are obtained, albeit under strong conditions. However, these results only relate to estimating marginal posterior distributions, rather than the full posterior.\\
\\ 
Sub-sampling methods are designed so that each iteration requires access to only a subset of the data. Exact approaches in this vein typically require subsets of the data of random size at each iteration. One approach is to construct unbiased estimators of point-wise evaluations of the target density using subsets of the data, which could then be embedded within the pseudo-marginal MCMC framework recently developed by \cite{as:ar09}. Unfortunately, the construction of such positive unbiased estimators is not possible in general \citep{as:jt15} and such methods often require both bounds on, and good analytical approximations of, the likelihood \citep{ijcai:ma15}.\\
\\
More promising practical results have been obtained by approximate sub-sampling approaches. These methods use subsamples of the data to estimate quantities such as acceptance probabilities \citep{arxiv:nfw12,icml:kcw14,icml:bdh14}, or the gradient of the posterior, that are used within MCMC algorithms. These estimates are then used in place of the true quantities. Whilst this can lead to increases in computational efficiency, the resulting algorithms no longer target the true posterior. The most popular of these algorithms is the stochastic gradient Langevin dynamics algorithm of \cite{icml:wt11}. This approximately samples a Langevin diffusion which has the posterior as its stationary distribution. To do this requires first approximating the continuous-time diffusion by a discrete-time Markov process, and then using sub-sampling estimates of the gradient of the posterior within the dynamics of this discrete-time process. This idea has been extended to approximations of other continuous-time dynamics that target the posterior \citep{icml:abw12,icml:cfg14,nips:mcf15}.\\
\\
Within these sub-sampling methods it is possible to tune the subsample size, and sometimes the algorithm's step-size, so as to control the level of approximation. This leads to a trade-off, whereby increasing the computational cost of the algorithm can lead to samplers that target a closer approximation to the the true posterior. There is also substantial theory quantifying the bias in, say, estimates of posterior means, that arise from these methods \citep{jmlr:ttv16,jmlr:vzt16,nips:cdc15,aistats:hz17,Dalalyan:2017}, and how this depends on the subsample size and step-size. However, whilst they often work well in practice it can be hard to know just how accurate the results are for any given application. Furthermore, many of these algorithms still have a computational cost that increases linearly with data size \citep{arxiv:bdh15,arxiv:ndhv17,sc:bffn19}.\\
\\
The approach to the problem of big data proposed here is a significant departure from the current literature. Rather than building our methodology upon the stationarity of appropriately constructed Markov chains,  a novel approach based on the {\em quasi-limiting} distribution of suitably constructed stochastically weighted diffusion processes is developed. A {\em quasi-limiting distribution} for a Markov process $X$ with respect to a Markov stopping time $\zeta $ is the limit of the distribution of $X_t\mid \zeta>t$ as $t \to \infty $, and such distributions are automatically \emph{quasi-stationary distributions}, see \cite{bk:qsd}; this concept is completely unrelated to the popular area of Quasi-Monte Carlo.  These {\em Quasi-Stationary Monte Carlo (QSMC) methods} developed can be used for a broad range of Bayesian problems (of a similar type to MCMC) and exhibit interesting and differing algorithmic properties. The QSMC methods developed are {\em exact} in the same (weak) sense as MCMC, in that they give the correct (quasi-)limiting distribution. There are a number of different possible implementations of the theory which open up interesting avenues for future research, in terms of  branching processes, by means of stochastic approximation methods, or (as outlined in this paper) SMC methods. We note that the use of continuous-time SMC and related algorithms to obtain approximations of large time limiting distributions of processes conditioned to remain alive has also been explored in settings in which a quantity of interest admits a natural representation of this form (see \cite{esaim:dm03,rousset2006}, and related work in the physics literature, such as \cite{giardina2011} and references therein); a substantial difference between these and the present work is that the QSMC methods described here construct a process for which a quite general distribution of interest is the quasi-stationary distribution and entirely avoid time-discretisation errors. One particularly interesting difference between our class of Monte Carlo algorithms and MCMC is that QSMC methods allow us to circumvent entirely the Metropolis-Hastings type accept/reject steps, while still retaining theoretical guarantees that the correct limiting target distribution is recovered. In the case of big data problems, this removes one of the fundamental $\mathcal{O}(n)$ bottlenecks in computation.\\
\\
Quasi-Stationary Monte Carlo methods can be applied in big data contexts by using a novel sub-sampling approach. We call the resulting algorithm the \textit{Scalable Langevin Exact Algorithm (ScaLE)}. The name refers to the `Langevin'  diffusion which is used in the mathematical construction of the algorithm, although it should be emphasised that it is not explicitly used in the algorithm itself. As shown in \secref{s:scale}, the approach to sub-sampling adopted here can potentially decrease the computational complexity of each iteration of QSMC to be $\mathcal{O}(1)$. Furthermore, for a rejection sampler implementation of QSMC, the use of sub-sampling introduces no additional error --- as the rejection sampler will sample from the same stochastic process, a killed Brownian motion, regardless of whether sub-sampling is used or not. There can be a computational cost of using sub-sampling, as the number of rejection-sampler iterations needed to simulate the killed Brownian motion for a given time interval will increase. However, this paper will show that by using control variates \citep{arxiv:bdh15} to reduce the variability of sub-sampling estimators of features of the posterior, the ongoing algorithm computational cost  will be
$\mathcal{O}(1)$. Constructing the control variates involves a pre-processing step whose cost is $\mathcal{O}(n)$ (at least in the case of posterior contraction at rate $n^{-1/2}$) but after this pre-processing step the resulting cost of ScaLE per effective sample size can be $\mathcal{O}(1)$. The importance of using control variates to get a computational cost that is sub-linear in $n$ is consistent with other recent work on scalable Monte Carlo methods \citep{aistats:hz17,arxiv:bfr16,arxiv:qmcr16,nips:drwp16,arxiv:ndhv17,sc:bffn19}.\\
\\
The next section presents the main result that motivates  development of quasi-stationary Monte Carlo. The following sections then provide detail on how to implement QSMC algorithms in practice, and how and why they are amenable to use with sub-sampling ideas. For clarity of presentation, much of the technical and algorithmic detail has been suppressed, but can be found in the appendices.


\section{Quasi-stationary Monte Carlo} \label{s:qsmc}

Given a target density $\pi $ on $\mathbbm{R}^d$, traditional (i.e. Metropolis-Hastings type) MCMC proposes at each iteration from Markov dynamics with proposal density $q(\bx,\by)$, `correcting' its trajectory by either accepting the move with probability
\begin{align}
\alpha (\bx ,\by ) = \min \left\{
1, {\pi (\by ) q(\by , \bx ) \over \pi (\bx ) q(\bx , \by ) }
\right\},
\end{align}
or rejecting the move and remaining at state $\bx $. In {\em quasi-stationary Monte Carlo}, rather than rejecting a move and staying at $\bx $, the algorithm {\em kills} the trajectory entirely, according to probabilities which relate to the target density.\\
\\
Simulation of a Markov process with killing inevitably leads to death of the
process. Thus it is natural to describe the long-term behaviour of the process through its conditional distribution given that the process is still alive. The limit of this distribution is called the quasi-stationary distribution (see, for example, \cite{bk:qsd}). The idea of quasi-stationary Monte Carlo is to construct a Markov process whose quasi-stationary distribution is the distribution, $\pi(\bx)$, from which the user wishes to sample from. Simulations from such a process can then be used to approximate expectations with respect to $\pi(\bx)$ just as in MCMC.\\ 
\\
Although in principle QSMC can be used with any Markov process, this paper will work exclusively with killed Brownian motion as it has a number of convenient properties that can be exploited. Therefore let $\{\mathbf{X}_t, t\ge 0\}$ denote $d$-dimensional Brownian motion initialised at $\mathbf{X}_0=\bx _0$. Suppose $\kappa (\bx )$ denotes a non-negative hazard rate at which the Brownian motion is killed when it is in state $\bx $, and let $\zeta $ be the killing time itself. Finally define
\begin{equation}
\label{e:mut} 
\mu _t(d\bx ) := \mathbbm{P}(\bX_t \in d\bx \mid \zeta >t ),
\end{equation}
the distribution of $\bX_t$ given that it has not yet been killed. The limit of this distribution as $t\rightarrow\infty$ is the quasi-stationary distribution of the killed Brownian motion.\\
\\
The aim will be to choose $\kappa $ in such a way that $\mu _t$ converges to $\pi$, and with this in mind, we introduce the function $\phi :\mathbbm{R}^d\to \mathbbm{R}$
\begin{align}
\label{e:phidef}
\phi(\bx) :=  {
\| \nabla \log \pi (\bx )  \| ^2
+ \Delta \log \pi (\bx ) \over 2} = {\Delta \pi(\bx) \over 2 \pi(\bx) },
\end{align}
where $\|\cdot\|$ denotes the usual Euclidean norm and $\Delta $ the Laplacian operator. By further imposing the condition
\begin{cond} \label{cond:phi}
There exists a constant $\Phi > -\infty$ such that $\Phi \leq \phi(\mathbf{u})$ $\forall\mathbf{u}\in\mathbbm{R}^d$.
\end{cond}
\noindent the following result can be proved:
\begin{theorem}
\label{thm:qsd}
Under the regularity conditions (\ref{eq:Q0}) and (\ref{eq:Q1}) in \apxref{apx:proofthm1}, suppose that \conref{cond:phi} holds and set
\begin{equation}
\label{e:kappadef}
\kappa (\bx ) :=  \phi (\bx ) - \Phi \geq 0,
\end{equation}
then it follows that $\mu _t$ converges in $L^1$ and pointwise to $\pi $.
\proof See \apxref{apx:proofthm1}. \endproof
\end{theorem}
\noindent Note that the regularity conditions in \apxref{apx:proofthm1} are largely technical smoothness and other weak regularity conditions common in stochastic calculus. On the other hand \conref{cond:phi} is necessary for us to be able to construct quasi-stationary Monte Carlo methods. However, since non-pathological densities on $\mathbbm{R}^d$ are typically convex in the tails, the second identity in (\ref{e:phidef}), demonstrates that \conref{cond:phi} is almost always satisfied in real examples.\\
\\
\thmref{thm:qsd} can be exploited for statistical purposes by noting that for some sufficiently large $t^*$, $\mu_t\approx \pi$ for $t>t^*$. Thus, given samples from $\mu_t$ for $t>t^*$, one would have an (approximate) sample from $\pi$. This is analogous to MCMC, with $t^*$ being the burn-in period; the only difference being the need to simulate from the distribution of the process conditional upon it not having died.\\
\\
The next two sections describe how to simulate from $\mu_t$. Firstly a  description of how to simulate killed Brownian motion process exactly in continuous-time is provided. A na\"{i}ve approach to sample from $\mu_t$, is to simulate independent realisations of this killed Brownian motion, and use the values at time $t$ of those processes which have not yet died by time $t$. In practice this is impracticable, as the probability of survival will, in general, decay exponentially with $t$. To overcome this  sequential Monte Carlo methods are employed. \\
\\
Both these two steps introduce additional challenges not present within standard MCMC. Thus a natural question is: why use quasi-stationary Monte Carlo at all? This is addressed this in \secref{s:scale} where it is shown that  simulating the killing events can be carried out using just subsamples of data.  In fact  subsamples of size 2 can be used without introducing any approximation into the dynamics of the killed Brownian motion.


\section{Implementing QSMC} \label{s:emcfd}


\subsection{Simulating Killed Brownian Motion} \label{s:killedBM}

\thmref{thm:qsd} relates a target distribution of interest
 to the quasi-stationary distribution of an appropriate killed Brownian motion. To be able to simulate from this quasi-stationary distribution it is necessary to be able to simulate from killed Brownian motion. 
 \\
\\
To help get across the main ideas, first consider the case where the killing rate, $\kappa(\bx)$, is bounded above by some constant, $\Kappa$ say. In this case it is possible to use thinning (see for example, \cite{bk:pp}) to simulate the time at which the process will die. This involves simulating the Brownian motion independently of a Poisson process with rate $\Kappa$. Each event of the Poisson process is a potential death event, and an appropriate Bernoulli variable then determines whether or not the death occurs. For an event at time $\xi$ the probability that death occurs depends on the state of the Brownian motion at time $\xi$, and is equal to $\kappa(\bx_{\xi})/\Kappa$. Thus to simulate the killed Brownian motion to time $t$ the first step is to simulate all events in the Poisson process up to time $t$. Then by considering the events in time-order, it is straightforward to simulate the Brownian motion at the first event-time and as a result determine whether death occurs. If death does not occur, the next event-time can be considered. This is repeated until either the process dies or the process has survived the last potential death event in $[0,t]$. If the latter occurs, Brownian motion can be simulated at time $t$ without any further conditions.\\
\\
This can be viewed as a rejection sampler to simulate from $\mu_t(\bx)$, the distribution of the Brownian motion at time $t$ conditional on it surviving to time $t$. Any realisation that has been killed is `rejected' and a realisation that is not killed is a draw from $\mu_t(\bx)$. It is easy to construct an importance sampling version of this rejection sampler. Assume there are $k$ events in the Poisson process before time $t$, and these occur at times $\xi_1,\ldots,\xi_k$.  The Brownian motion path is simulated at each event time and at time $t$. The output of the importance sampler is the realisation at time $t$, $\bx_t$, together with an importance sampling weight that is equal to the probability of the path surviving each potential death event,
\begin{align}
W_t := \prod_{i=1}^k \frac{\Kappa-\kappa(\bx_{\xi_i})}{\Kappa}. \nonumber
\end{align}
Given a positive lower bound on the killing rate, $\kappa(\bx)\geq \Kappa^\downarrow$ for all $\bx$, it is possible to improve the computational efficiency of the rejection sampler by splitting the death process into a death process of rate $\Kappa^\downarrow$ and one of rate $\kappa(\bx)-\Kappa^\downarrow$. Actual death occurs at the first event in either of these processes. The advantage of this construction is that the former death process is independent of the Brownian motion. Thus it is possible to first simulate whether or not death occurs in this process. If it does not we can then simulate, using thinning as above, a killed Brownian motion with rate $\kappa(\bx)-\Kappa^\downarrow$. The latter will have a lower intensity and thus be quicker to simulate. Using the importance sampling version instead, events in a Poisson process of rate $\Kappa-\Kappa^\downarrow$, $\xi_1,\ldots,\xi_k$ say, are simulated, and our realisation at time $t$ is assigned a weight
\begin{align}
W_t := \exp\{-\Kappa^\downarrow t\} \prod_{i=1}^k \frac{\Kappa- \kappa(\bx_{\xi_i})}{\Kappa-\Kappa^\downarrow}. \nonumber
\end{align}
This is particularly effective as the $\exp\{-\Kappa^\downarrow t\}$ is a constant which will cancel upon normalisation of the importance sampling weights.


\subsection{Simulating Killed Brownian Motion using Local Bounds}\label{s:langevin} 

The approach in \secref{s:killedBM} is not applicable in the absence of an upper bound on the killing rate. Even in situations where a global upper bound does exist, the resulting algorithm may be inefficient if this bound is large. Both of these issues can be overcome using local bounds on the rate. For this section we will work with the specific form of the killing rate in \thmref{thm:qsd}, namely $\phi(\bx)-\Phi$. The bounds used will be expressed in terms of bounds on $\phi(\bx)$.\\
\\
Given an initial value for the Brownian motion, $\bx_0$, define a hypercube which contains $\bx_0$. In practice this cube is defined to be centred on $\bx_0$ with a user-chosen side length (which may depend on $\bx_0$). Denote the hypercube by $\mathcal{H}_1$, and assume that upper and lower bounds, $U_{\bX}^{(1)}$ and $L_{\bX}^{(1)}$ respectively, can be found for $\phi(\bx)$ with $\bx\in \mathcal{H}_1$.  The thinning idea of the previous section can be used to simulate the killed Brownian motion whilst the process stays within $\mathcal{H}_1$. Furthermore it is possible to simulate the time at which the Brownian motion first leaves $\mathcal{H}_1$ and the value of the process when this happens (see \apxref{s:pathspace layer}). Thus our approach is to use our local bounds on $\phi(\bx)$, and hence on the killing rate, to simulate the killing process while $\bx$ remains in $\mathcal{H}_1$. If the process leaves $\mathcal{H}_1$ before $t$ it is then necessary to define a new hypercube, $\mathcal{H}_2$ say, obtain new local bounds on $\phi(\bx)$ for $\bx \in \mathcal{H}_2$ and repeat simulating the killing process using these new bounds until the process either first leaves the hypercube or time $t$ is reached.\\
\\
The details of this approach are now given, describing the importance sampling version which is used later --- though a rejection sampler can be obtained using similar ideas. The first step is to calculate the hypercube, $\mathcal{H}_1$, and the bounds $L_{\bX}^{(1)}$, $U_{\bX}^{(1)}$.  We then simulate the time and position at which $\bx$ first leaves $\mathcal{H}_1$. We call this the {\em layer information}, and denote it as $R_{\bX}^{(1)}=(\tau_1,\bx_{\tau_1})$. The notion of a layer for diffusions was formalised in \cite{b:pjr15}, and we refer the interested reader there for further details. Next the possible killing events on $[0,t \wedge \tau_1]$ are generated by simulating events of a Poisson process of rate $U_{\bX}^{(1)}-L_{\bX}^{(1)}$: $\xi_1,\ldots,\xi_k$ say. The next step involves simulating the values of the Brownian motion at these event times (the simulation of which is conditional on $R_{\bX}^{(1)}$ --- see  \apxref{ss:fptinter} and \algref{alg:qWqsim} for a description of how this can be done). An incremental importance sampling weight for this segment of time is given as
\begin{align}
W^{(1)}:=
\exp\left\{ -\left(L_{\bX}^{(1)}-\Phi  \right)\cdot(t \wedge \tau_1) \right\} \prod_{i=1}^k \frac{
U_{\bX}^{(1)}  -
\phi(\bx_{\xi_i})
}{U_{\bX}^{(1)}-L_{\bX}^{(1)}}.
\end{align}
If $\tau_1<t$, then  this process is repeated with a hypercube centred on $\bx_{\tau_1}$ until simulation to time $t$ has been achieved. This gives successive incremental weights $W^{(2)}, W^{(3)}, \ldots $.  A simulated value for the  Brownian motion at time $t$ is given, again simulated conditional on the layer information for the current segment of time, and an importance sampling weight that is the product of the incremental weights associated with each segment of time. At time $t$,  $J(t)$ incremental weights  have been simulated leading to the cumulative weight
\begin{align}
W_t = \prod_{j=1}^{J(t)} W^{(j)}. \label{eq:incremental weight}
\end{align}
Full algorithmic detail of the description above are given in \algref{alg:is-kbm}. In practice every sample $\bX_t$ will have an importance weight that shares a common constant of $\exp\{ \Phi t \}$ in (\ref{eq:incremental weight}). As such it is omitted from \algref{alg:is-kbm} and the weights are asterisked to denote this. It is straightforward to prove that this approach gives valid importance sampling weights in the following sense.
\begin{theorem}
\label{unbiasedwts}
For each $t\le T$ 
\begin{align}
\mathbbm{E}[ W_t \mid \bX[0,T]
] = e^{-\int_0^t \phi (X_s) ds} \nonumber
\end{align}
\end{theorem}
 \proof
First note that by direct calculation of its Doob-Meyer decomposition conditional on $\bX[0,T]$, $W_t e^{\int_0^t \phi (X_s) ds}$ is a martingale, see for example \cite{bk:cmbm}. Therefore $\mathbbm{E}[W_t|\bX[0,T]]e^{\int_0^t \phi (X_s) ds}=1$ and the result follows.
 \qed 
\begin{algorithm}[h]
	\caption{Importance Sampling Killed Brownian Motion (IS-KBM) Algorithm} \label{alg:is-kbm}
	\begin{enumerate}
	\item Initialise: Input initial value $\mathbf{X}_0$, and time interval length $t$. Set $i=1$, $j=0$, $\tau_0=0$, $w^*_0=1$.
	\item $R$: Choose hypercube $\mathcal{H}_i$ and calculate $L^{(i)}_\mathbf{X}$, $U^{(i)}_\mathbf{X}$. 
	Simulate layer information $R^{(i)}_{\mathbf{X}}\sim\mathcal{R}$ as per \apxref{s:pathspace layer}, obtaining $\tau_i, \bx_{\tau_i}$. \label{alg:is-kbm:layer}
	\item $E$: Simulate $E \sim \text{Exp}(U^{(i)}_\mathbf{X}-L^{(i)}_\mathbf{X})$. \label{alg:is-kbm:loop} 
	\item $\xi_j$: Set $j=j+1$ and $\xi_j =  (\xi_{j-1}+E)\wedge \tau_i \wedge t$. \label{alg:qsmc:scalesim2}
	\item $w^*_{\xi_j}$: Set $w^*_{\xi_j}=w^*_{\xi_{j-1}}\cdot\exp\{-L^{(i)}_\mathbf{X}[ \xi_j-\xi_{j-1}]\}$.
	\item $\mathbf{X}_{\xi_j}$: Simulate $\mathbf{X}_{\xi_j}\sim \left.\text{MVN}(\mathbf{X}_{\xi_{j-1}},(\xi_j-\xi_{j-1}))|R^{(i)}_\mathbf{X} \right.$ as per \apxref{ss:fptinter} and \algref{alg:qWqsim}.
	\item $\tau_i$: If $\xi_j=t$ then output $\bx_t$ and $w^*_t$. Otherwise, if $\xi_j=\tau_i$, set $i=i+1$, and return to \stepref{alg:is-kbm:layer}. Else set $w^*_{\xi_j}=w^*_{\xi_j}\cdot (U^{(i)}_\mathbf{X}-\phi(\mathbf{X}_{\xi_j}) )/ (U^{(i)}_\mathbf{X} -  L^{(i)}_\mathbf{X}) $ and return to \stepref{alg:is-kbm:loop}.\label{alg:is-kbm:w}
	\end{enumerate}
\end{algorithm}


\subsection{Simulating from the Quasi-stationary Distribution} \label{s:emcfd:smc} \label{s:convergence} \label{s:qsmc} 

In theory we can use our ability to simulate from $\mu_t(\bx)$, using either rejection sampling to simulate from the quasi-stationary distribution of our killed Brownian motion, or importance sampling to approximate this distribution. We would need to specify a `burn-in' period of length $t^*$ say, as in MCMC, and then simulate from $\mu_{t^*}(\bx)$. If $t^*$ is chosen appropriately these samples would essentially be draws from the quasi-stationary distribution. Furthermore we can propagate these samples forward in time to obtain samples from $\mu_t(\bx)$ for $t>t^*$, and these would, marginally, be draws from essentially the quasi-stationary distribution.\\
\\
However, in practice this simple idea is unlikely to work. We can see this most clearly with the rejection sampler, as the probability of survival will decrease exponentially with $t$ --- and thus the rejection probability will often be prohibitively large.\\
\\
There have been a number of suggested approaches to overcome the inefficiency of this na{\"i}ve approach to simulating from a quasi-stationary distribution (see for example  \cite{pre:dd05,mprf:gj13}, and the recent rebirth methodology of \cite{aap:bgz16}). Our approach is to use ideas from sequential Monte Carlo. In particular, we will discretise time into $m$ intervals of length $T/m$ for some chosen $T$ and $m$. Defining $t_i:=iT/m$ for $i=1,\ldots,m$, we use our importance sampler to obtain an $N$-sample approximation of $\mu_{t_1}(\bx)$; this will give us $N$ particles, that is realisations of $\bx_{t_1}$, and their associated importance sampling weights. We normalise the importance sampling weights, and calculate the empirical variance of these normalised weights at time $t_1$. If this is sufficiently large we resample the particles, by simulating $N$ times from the empirical distribution defined by the current set of weighted particles. If we resample, we assign each of the new particles a weight $1/N$.\\
\\ 
The set of weighted particles at time $t_1$ is then propagated to obtain a set of $N$ weighted particles at time $t_2$. The new importance sampling weights are just the weights at time $t_1$, prior to propagation, multiplied by the (incremental) importance sample weight calculated when propagating the particle from time $t_1$ to $t_2$. The above resampling procedure is applied, and this whole iteration is repeated until we have weighted particles at time $T$. This approach is presented as the \textit{Quasi-Stationary Monte Carlo (QSMC)} algorithm in \algref{alg:qsmc} in which ${N_{\text{eff}}}$ is the effective sample size of the weights \citep{jasa:klw94}, a standard way of monitoring the variance of the importance sampling weights within sequential Monte Carlo, and $N_{\text{th}}$ is a user chosen threshold which determines whether or not to resample. The algorithm outputs the weighted particles at the end of each iteration.
\begin{algorithm}[h]
	\caption{Quasi-Stationary Monte Carlo Algorithm (QSMC) Algorithm.} \label{alg:qsmc}
    \begin{enumerate}
	\item \textbf{Initialisation Step ($i=0$)}
	\begin{enumerate}
	\item Input: Starting distribution, $f_{\bx_0}$, number of particles, $N$, and  set of $m$ times $t_{1:m}$.
	\item $\mathbf{X}^{(\cdot)}_0$: For $k$ in $1$ to $N$ simulate $\mathbf{X}^{(1:N)}_{t_0}\sim f_{\bx_0}$ and set $w^{(1:N)}_{t_0}=1/N$.
	\end{enumerate}
	\item \textbf{Iterative Update Steps ($i=i+1$ while $i\leq m$)}
	\begin{enumerate}
	\item $N_{\text{eff}}$: If ${N_{\text{eff}}}\leq N_{\text{th}}$ then for $k$ in $1$ to $N$ resample $\mathbf{X}^{(k)}_{t_{i-1}} \sim \tilde{\pi}^N_{t_{i-1}}$, the empirical distribution defined by the current set of weighted particles, and set $w^{(k)}_{t_{i-1}} = 1/N$.
	\item For $k$ in $1$ to $N$,
	\begin{enumerate}
	\item $\mathbf{X}^{(\cdot)}_{t_i}$: Simulate $\mathbf{X}^{(k)}_{t_i} | \mathbf{X}^{(k)}_{t_{i-1}}$ along with un-normalised weight increment $w^*_{t_i-t_{i-1}}$ as per \algref{alg:is-kbm}.\label{alg:qsmc:scalesim1} 
	\item ${w'}^{(\cdot)}_{\!\!t_i}$: Calculate un-normalised weights, ${w'}^{(k)}_{\!\!t_i} = w^{(k)}_{t_{i-1}}\cdot w^{*}_{t_i-t_{i-1}}$. \label{alg:qsmc:accrej1}
	\end{enumerate}
	\item $w^{(\cdot)}_{t_i}$: For $k$ in $1$ to $N$ set $w^{(k)}_{t_i} = {w'}^{(k)}_{\!\!t_i}/\sum^N_{l=1} {w'}^{(l)}_{\!\!t_i}$.
	\item $\tilde{\pi}^N_{t_i}$: Set $\tilde{\pi}^N_{t_i}(\!\ud \mathbf{x}) := \sum^N_{k=1} w^{(k)}_{t_i}\cdot\delta_{\mathbf{X}^{(k)}_{t_i}}(\!\ud \mathbf{x})$.
	\end{enumerate}	
	\end{enumerate}
\end{algorithm}
\\\\ \noindent Given the output from \algref{alg:qsmc},  the target distribution $\pi$ can be estimated as follows. 
After choosing a burn-in time, $t^*(\in(t_0,\ldots{},t_m))$, sufficiently large to provide reasonable confidence that quasi-stationarity has been `reached'. The approximation to the law of the killed process is then simply the weighted occupation measures of the particle trajectories in the interval $[t^*,T]$. More precisely, using the output of the QSMC algorithm,
\begin{align}
\pi(\!\ud \mathbf{x}) \approx \hat{\pi}(\!\ud \mathbf{x}) 
& := \dfrac{1}{m(T-t^*)/T}\sum^m_{i=m(T-t^*)/T} \sum^N_{k=1}w^{(k)}_{t_i}\cdot \delta_{\mathbf{X}^{(k)}_{t_i}}(\!\ud \mathbf{x}). \label{eq:occupation2}
\end{align}
For concreteness, for a suitable $L^1(\pi )$ function, $g$, the Monte Carlo estimator can simply be set to,
\begin{equation}\label{eq:occupation3}
{\widehat {\pi (g)} }= {1\over m(T-t^*)/T} \sum_{i=m(T-t^*)/T}^m\sum_{k=1}^N 
w^{(k)}_{t_i} \cdot g\left({ \mathbf{X}^{(k)}_{t_i}}\right).
\end{equation}
The general ($g$-specific) theoretical effective sample size (ESS) is given by $\text{Var}_\pi\,g/\text{Var}\,{\widehat {\pi (g)} }$. Practical approximation of ESS is discussed in \apxref{app:ESS}.

\FloatBarrier
\section{Sub-sampling} \label{s:scale} 

We now return to the problem of sampling from the posterior in a big data setting and will assume we can write the target posterior as
\begin{align}
\pi(\bx) (=\pi _n(\bx) ) \propto \prod_{i=0}^n f_i(\bx), 
\end{align}
where $f_0(\bx)$ is the prior and $f_1(\bx),\ldots,f_n(\bx)$ are likelihood terms. Note that to be consistent with our earlier notation $\bx$ refers to the parameters in our model. The assumption of this factorisation is quite weak and includes many classes of models exhibiting various types of conditional independence structure.\\
\\
It is possible to  sample from this posterior using \algref{alg:qsmc} by choosing $\phi(\bx)$, and hence $\kappa(\bx)$, which determines the death rate of the killed Brownian motion, as defined in (\ref{e:phidef}) and (\ref{e:kappadef}) respectively. In practice this will be computationally prohibitive as at every potential death event we determine acceptance by evaluating $\phi(\bx)$, which involves calculating derivatives of the log-posterior, and so requires accessing the full data set of size $n$. However, it is easy to  estimate $\phi(\bx)$ unbiasedly using sub-samples of the data as the log-posterior is a sum over the different data-points. Here we show that we can use such an unbiased estimator of $\phi(\bx)$ whilst still simulating the underlying killed Brownian motion exactly.


\subsection{Simulating Killed Brownian Motion with an Unbiased Estimate of the Killing Rate} \label{s:poissonintro}

To introduce the proposed approach we begin by assuming we can simulate an auxiliary random variable $A\sim\mathcal{A}$, and (without loss of generality) construct a positive unbiased estimator, $\tilde{\kappa}_A(\cdot)$, such that
\begin{align} \label{eq:kappaunbiased}
 \mathbbm{E}_{\mathcal{A}}\left[\tilde{\kappa}_A(\cdot)\right] = \kappa(\cdot).
\end{align}
The approach relies on the following simple result which is stated in a general way as it is of independent interest for simulating from events of probability which are expensive to compute, but that admit a straightforward unbiased estimator. Its proof is trivial and will be omitted.
\begin{proposition} \label{prop: p-coin}
Let $0\le p\le 1$, and suppose that $P$ is a random variable with $\mathbbm{E}(P)=p$ and $0\le P\le 1$ almost surely. Then if $u\sim \U[0,1]$, the event
$\{ u\le P\}$ has probability $p$.
\end{proposition}
\noindent We now adapt this result to our setting, noting that the randomness obtained by direct simulation of a $p$-coin, and that using \propref{prop: p-coin}, is indistinguishable.\\
\\
Recall that in \secref{s:killedBM} in order to simulate a Poisson process of rate $\kappa$, Poisson thinning was used. The initial step is to first find, for the Brownian motion trajectory constrained to the hypercube $\mathcal{H}$, a constant $ \Kappa_\bX \in \mathbbm{R}_+$ such that $\forall \bx \in \mathcal{H}$, $\kappa(\bx) \leq  \Kappa_\bX$ holds. Then a dominating Poisson process of rate $\Kappa_\bX$ is simulated to obtain potential death events, and then in sequence each potential death event is accepted or rejected. A single such event, occurring at time $\xi$ will be accepted as a death with probability $\kappa(\bx_\xi)/\Kappa_\bX$.\\
\\
An equivalent formulation would simulate a Poisson process of rate $\kappa$ using a dominating Poisson process of higher rate $\tilde{\Kappa}_\bX\geq \Kappa_\bX$. This is achieved by simply substituting $\Kappa_\bX$ for $\tilde{\Kappa}_\bX$ in the argument above. However, the penalty for doing this is an increase in the expected computational cost by a factor of $\tilde{\Kappa}_\bX/\Kappa_\bX$ -- therefore it is reasonable to expect to have a larger number of potential death events, each of which will have a smaller acceptance probability.\\
\\
Now, suppose for our unbiased estimator $\tilde{\kappa}_A$ it is possible to identify some $ \tilde{\Kappa}_\bX \in \mathbbm{R}_+$ such that for $\mathcal{A}$-almost all  $A$, and all $\bx\in \mathcal{H}$, $0 \leq \tilde{\kappa}(\bx) \leq \tilde{\Kappa}_\bX$. Noting from (\ref{eq:kappaunbiased}) that we have an unbiased $[0,1]$-valued estimator of the probability of a death event in the above argument (i.e. $\mathbbm{E}_{\mathcal{A}}[\tilde{\kappa}_A(\bx)/ \tilde{\Kappa}] = \kappa(\bx)/\tilde{\Kappa}$), and by appealing to \propref{prop: p-coin}, another (entirely equivalent) formulation of the Poisson thinning argument above is to use a dominating Poisson process of rate $\tilde{\Kappa}_\bX$, and determine acceptance or rejection of each potential death event by simulating $A\sim \mathcal{A}$ and accepting with probability $\tilde{\kappa}_A(\bx_\xi)/\tilde{\Kappa}$ (instead of $\kappa(\bx_\xi)/\tilde{\Kappa}$).\\
\\
In the remainder of this section we exploit this extended construction of Poisson thinning (using an auxiliary random variable and unbiased estimator), to develop a scalable alternative to the QSMC approach introduced in \algref{alg:qsmc}. The key idea in doing so is to find an auxiliary random variable and unbiased estimator which can be simulated and evaluated without fully accessing the data set, while ensuring the increased number of evaluations necessitated by the ratio $\tilde{\Kappa}_\bX/\Kappa_\bX\geq 1$ does not grow too severely.


\subsection{Constructing a scalable replacement estimator} \label{s:repest}

Noting from (\ref{e:phidef}) and (\ref{e:kappadef}) that the selection of $\kappa(\bx)$ required to sample from a posterior $\pi(\bx)$ is determined by $\phi(\bx)$, in this section we focus on finding a practical construction of a scalable unbiased estimator for $\phi(\bx)$. Recall that, 
\begin{align}
\phi(\mathbf{x}) := (\| \nabla\log\pi(\mathbf{x})\|^2 + \Delta
\log\pi(\mathbf{x}))/2, 
\end{align}
and that as in \algref{alg:qsmc}, whilst staying within hypercube $\mathcal{H}_i$, it is possible to find constants ${L}^{(i)}_\bX$ and ${U}^{(i)}_\bX$ such that ${L}^{(i)}_\bX\leq \phi(\bx) \leq {U}^{(i)}_\bX$. As motivated by \secref{s:poissonintro}, it is then possible to construct an auxiliary random variable $A\sim\mathcal{A}$, and an unbiased estimator $\phi_A$ such that 
\begin{align} \label{eq:phiunbiased}
 \mathbbm{E}_{\mathcal{A}}\left[\phi_A(\cdot)\right] = \phi(\cdot),
\end{align}
and to determine constants $\tilde{U}^{(i)}_\bX\geq U^{(i)}_\bX$ and $\tilde{L}^{(i)}_\bX\leq L^{(i)}_\bX$ such that within the same hypercube we have $\tilde{L}^{(i)}_\bX\leq \tilde{\phi}_A(\bx) \leq \tilde{U}^{(i)}_\bX$. To ensure the validity of our QSMC approach, as justified by \thmref{thm:qsd} in \secref{s:qsmc}, it is necessary to substitute \conref{cond:phi} with the following (similarly weak) condition:
\begin{cond} \label{cond:phitilde}
There exists a constant $\tilde{\Phi} > -\infty$ such that $\tilde{\Phi} \leq \tilde{\phi}_A(\mathbf{u})$ for $\mathcal{A}$-almost every $A$, $\forall\mathbf{u}\in\mathbbm{R}^d$.
\end{cond}
\noindent To ensure practicality and scalability it is crucial to focus on ensuring that the ratio
\begin{align}
\dfrac{\tilde{\lambda}}{\lambda}  = \dfrac{\tilde{U}^{(i)}_\bX-\tilde{L}^{(i)}_\bX}{{U}^{(i)}_\bX-{L}^{(i)}_\bX} \ ,
\label{eq:lamba def}
\end{align}
where $\tilde{\lambda}:= \tilde{U}^{(i)}_\bX-\tilde{L}^{(i)}_\bX$
does not grow too severely with the size of the data set (as this determines the multiplicative increase in the rate, and hence increased inefficiency, of the dominating Poisson process required within \algref{alg:qsmc}). To do this, our approach develops a tailored control variate, of a similar type to that which has since been successfully used within the concurrent work of two of the authors on MCMC (see \cite{arxiv:bfr16}).\\
\\
To implement the control variate estimator, it is first necessary to find a point close to a mode of the posterior distribution $\pi$, denoted by $\hat{\mathbf{x}}$. In fact for the scaling arguments to hold,  $\hat{\mathbf{x}}$ should be within $\mathcal{O}(n^{-1/2})$ of the true mode, and achieving this is a less demanding task than actually locating the mode. Moreover we note that this operation is only required to be done once, and not at each iteration, and so can be done fully in parallel. In practice it would be possible to use a stochastic gradient optimisation algorithm to find a value close to the posterior mode, and we recommend then starting the simulation of our killed Brownian motion from this value, or from some suitably chosen distribution centred at this value. Doing this substantially reduces the burn-in time of our algorithm. In the following section we describe a simpler method applicable when two passes of the full data set can be tolerated in the algorithm's initialisation.\\
\\
%
Addressing scalability for multi-modal posteriors is a more challenging problem, and goes beyond what is addressed in this paper, but is of significant interest for future work. We do, however, make the following remarks.
In the presence of multi-modality, stochastic gradient optimisation schemes may converge to the wrong mode. This is still good enough
as long as possible modes are separated by a distance which is $\mathcal{O}(n^{-1/2})$; when separate modes which are separated by more than $\mathcal{O}(n^{-1/2})$ exist, an interesting option would be to adopt multiple control variates.
\\
\\
Remembering that $\log \pi(\bx)=\sum_{i=0}^n \log f_i(\bx)$ and letting $\mathcal{A}$ be the law of $I\sim \U\{0, \ldots ,n\}$, our control variate estimator is constructed thus
\begin{align}
\mathbbm{E}_{\mathcal{A}}\big[\underbrace{(n+1)\cdot\left[\nabla\log f_I(\mathbf{x})-\nabla\log f_I(\hat{\mathbf{x}})\right]}_{=:\tilde{\alpha}_I(\mathbf{x})}\big] = \underbrace{\nabla\log\pi(\mathbf{x})-\nabla\log\pi(\hat{\mathbf{x}})}_{=:\alpha(\mathbf{x})}. \label{eq:alphatilde}
\end{align}
As such,  $\phi(\mathbf{x})$ can be re-expressed as
\begin{align}
\phi(\mathbf{x}) 
& = (\alpha(\mathbf{x})^{T}\left(2\nabla\log\pi(\hat{\mathbf{x}})
+ \alpha(\mathbf{x})\right) + \text{div\,}\alpha(\mathbf{x}))/2 + C,
\end{align}
where $C := \| \nabla\log\pi(\hat{\mathbf{x}})\|^2 /2
+  \Delta \log\pi(\hat{\mathbf{x}})/2$ is a constant. Letting $\mathcal{A}$ now be the law of $I,J\overset{\text{iid}}{\sim} \U\{0, \ldots ,n\}$ the following unbiased estimator of $\phi$ can be constructed,
\begin{align}
\mathbbm{E}_{\mathcal{A}}\big[\underbrace{(\tilde{\alpha}_I(\mathbf{x}))^{T}\left(2\nabla\log\pi(\hat{\mathbf{x}})+ \tilde{\alpha}_J(\mathbf{x})\right) + 
\text{div\,}\tilde{\alpha}_I(\mathbf{x}))/2 + C}_{=:\tilde{\phi}_A(\mathbf{x}) }\big] = \phi(\mathbf{x}). \label{eq:subsamplelambda}
\end{align}
The estimators $\tilde{\alpha}_I(\mathbf{x})$ and $\tilde{\phi}_A(\mathbf{x})$ are nothing more than classical {\em control variate} estimators, albeit in a fairly elaborate setting, and henceforth we shall refer to these accordingly.

The construction of the estimator requires evaluation of the constants $\nabla\log\pi(\hat{\mathbf{x}})$ and $\Delta
\log\pi(\hat{\mathbf{x}})$. Although both are $\mathcal{O}(n)$ evaluations they only have to be computed once, and furthermore, as mentioned above,  can be calculated entirely in parallel. \\
\\
The unbiased estimators $\tilde{\alpha}_I(\mathbf{x})$ and $\tilde{\phi}_A(\mathbf{x})$ use (respectively) single and double draws from $\{1,\ldots ,n\}$ although  it is possible to replace these by averaging over multiple draws (sampled with replacement), although this is not studied theoretically in the present paper and is exploited only in Section~\ref{s:examp:2} of the empirical study. \\
\\
Embedding our sub-sampling estimator described above within the QSMC algorithm of \secref{s:qsmc}, results in \algref{alg:scale} termed the \textit{Scalable Langevin Exact algorithm (ScaLE)}.  A similar modification could be made to the rejection sampling version, R-QSMC, which was discussed in \secref{s:qsmc} and detailed in \apxref{apx:rqsmc}. This variant  is termed the \textit{Rejection Scalable Langevin Exact algorithm (R-ScaLE)} and full algorithmic details are provided in \apxref{apx:r-scale}. 
\begin{algorithm}[h]
	\caption{The ScaLE Algorithm (as per \algref{alg:qsmc} unless stated otherwise).} \label{alg:subscale3} \label{alg:scale}
	\begin{enumerate}
    \item[0.] Choose $\hat{\mathbf{x}}$ and compute $\nabla\log\pi(\hat{\mathbf{x}})$, $
    \Delta \log\pi(\hat{\mathbf{x}})$.
        \item[\ref{alg:qsmc:scalesim1}.] On calling \algref{alg:is-kbm},
    \begin{enumerate}
    \item Replace $L^{(i)}_{\mathbf{X}},U^{(i)}_{\mathbf{X}}$ in \stepref{alg:is-kbm:layer} with $\tilde{L}^{(i)}_{\mathbf{X}},\tilde{U}^{(i)}_{\mathbf{X}}$.
    \item Replace \stepref{alg:is-kbm:w} with: $\tau_i$: If $\xi_j=\tau_i$, set $i=i+1$, and return to \stepref{alg:is-kbm:layer}. Else simulate $A_j=(I_j,J_j)$, with $I_j,J_j\overset{\text{iid}}{\sim} \U\{0, \ldots ,n\}$, and set $w^*_{\xi_j}=w^*_{\xi_j}\cdot (\tilde{U}^{(i)}_\mathbf{X}-\tilde{\phi}_{A_j}(\mathbf{X}_{\xi_j}) )/ (\tilde{U}^{(i)}_\mathbf{X} -  \tilde{L}^{(i)}_\mathbf{X})$ (where $\tilde{\phi}_{A_j}$ is defined as in (\ref{eq:subsamplelambda})) and return to \algstref{alg:is-kbm}{alg:is-kbm:loop}.
    \end{enumerate}
	\end{enumerate}
\end{algorithm}

\subsection{Implementation Details} \label{s:initproc}
In this section we detail some simple choices of the various algorithmic parameters which lead to a concrete implementation of the ScaLE algorithm. These choices have been made on the bases of parsimony and convenience and are certainly not optimal.\\
\\
In practice, we are likely to want to employ a suitable preconditioning transformation: $\bX'=\precon^{-1}\bX$ before applying the algorithm in order to roughly equate scales for different components. If we did not do this, it is likely that some components would mix particularly slowly.
Obtaining a suitable $\hat{\mathbf{x}}$ and $\precon$ is important. 
One concrete approach, and that used throughout our empirical study except where otherwise stated, is as follows. Divide a data set into a number of batches which are small enough to be processed using standard maximum likelihood estimation approaches and estimate the MLE and observed Fisher information for each batch; $\hat{\mathbf{x}}$ can then be chosen to be the mean of these MLEs and $\precon^{-1}$ to be a diagonal matrix with elements equal to the square root of the sum of the diagonal elements of the estimated information matrices. Better performance would generally be obtained using a non-diagonal matrix, but this serves to illustrate a degree of robustness to the specification of these parameters. The constants required within the control variate can then be evaluated. For a given hypercube, $\mathcal{H}$, a bound, $\widetilde{K}_{\mathbf{X}}$, on the dominating Poisson process intensity can then be obtained by simple analytic arguments facilitated by extending that hypercube to include $\hat{\mathbf{x}}$ and obtaining bounds on the modulus of continuity of $\widetilde{\phi}_A$. In total, two passes of the full dataset are required to obtain the necessary algorithmic parameters and to fully specify the control variate.\\
\\
As discussed in \secref{s:qsmc}, it is necessary to choose an execution time, $T$, for the algorithm and an auxiliary mesh ($t_0:=0,t_1,\ldots{},t_m:=T$) on which to evaluate $g$ in the computation of the QSMC estimator (\ref{eq:occupation3}). Note that within the algorithm the particle set is evolving according to killed Brownian motion with a preconditioning matrix $\precon^{-1}$ chosen to approximately match the square root of the information matrix of the target posterior. As such, $T$ should be chosen to match the time taken for preconditioned Brownian motion to explore such a space, which in the examples considered in this paper ranged from $T\approx 1$ to $T\approx 100$. The number of temporal mesh points, $m$, was chosen with computational considerations in mind --- increasing $m$ increases the cost of evaluating the estimator and leads to greater correlation between the particle set at consecutive mesh points, but ensures when running the algorithm on a multiple user cluster that the simulation is periodically saved and reduces the variance of the estimator. As the computational cost of the algorithm is entirely determined by the bounds on the discussed modulus of continuity of $\tilde{\phi}_A$, in each of the examples we later consider our mesh size was loosely determined by the inverse of this quantity and ranged from $(t_i-t_{i-1})\approx 10^{-3}$ to $(t_i-t_{i-1})\approx 10^{-6}$.\\
\\
The initial distribution $f_{\mathbf{x}_0}$ is not \emph{too} critical, provided that it is concentrated reasonably close (within a neighbourhood of size $\mathcal{O}(n^{-1/2})$) to the mode of the distribution. The stability properties of the SMC implementation ensure that the initial conditions will be forgotten (see Chapter 7 of \citet{bk:fk} for a detailed discussion). The empirical results presented below were obtained by choosing as $f_{\mathbf{x}}$, either a singular distribution concentrated at $\hat{\mathbf{x}}$ or a normal distribution centred at that location with a covariance matrix matching $\precon \precon^T$; results were found to be insensitive to the particular choice.


\section{
Complexity of ScaLE} \label{s:scalecomplexity}


The computational cost of ScaLE will be determined by two factors: the speed at which $\mu _t$ approaches $\pi $ and the computational cost of running the algorithm per unit algorithm time. Throughout the exposition of this paper, the proposal process is simple Brownian motion. Due to posterior contraction, as $n$ grows this proposal Brownian motion moves increasingly rapidly through the support of $\pi $. On the other hand, as $n$ grows, killing rates will grow. In this subsection we shall explore in detail how the computational cost of ScaLE varies with $n$ (its complexity) while bringing out explicitly the delicate link to the rate of posterior contraction and the effect of the choice of ${\hat\bx}$.

We start by examining the speed of convergence of $\mu _t$ and in particular its dependence on posterior contraction. Being more explicit about posterior contraction, we say that $\{\pi _n\}$ are $\mathcal{O}(n^{-\eta/2})$ 
or have {\em contraction rate $\eta /2$} for some $\eta > 0$ to a limit $\bx _0$ if for all $\epsilon >0$ there exists $K>0$ such that when $\bX _n \sim \pi _n$,
$
{\Bbb P} (|\bX_n - \bx _0|>Kn^{-\eta /2})<\epsilon $.
It is necessary to extend the definition of $\mu_t$ to a setting where $n$ increases, hence define
\begin{equation}
\label{e:mute} 
\mu _t^{n,{\mathbf{u}}}(d\bx ) := \mathbbm{P}(\bX_t \in d\bx \mid \zeta >t , \bX_0=\bx_0+n^{-\eta /2} {\mathbf{u}}).
\end{equation}
Since we are dealing with Markov processes that are essentially never uniformly ergodic, it is impossible to control convergence times uniformly. The specification of the initial value as
$\bX_0=\bx_0+n^{-\eta /2} {\mathbf{u}}$, which, as $n$ increases, remains close to the centre of the posterior as specified through the contraction rate, goes as far as we can before incurring additional computational costs for bad starting values.

Set $$
T_{n,{\mathbf{u}},\epsilon} = \inf\{t\ge 0;\ 
\|  \mu _t^{n,{\mathbf{u}}} - \pi _n\| <\epsilon
\}
$$
where $\| \cdot  \| $ represents total variation distance.
It will be necessary to make the following technical assumption. For all $\epsilon ,K >0$ 
\begin{equation}
\label{e:unifqsconv}
\limsup_{n\to \infty } 
\sup_{|\mathbf{u} |<K }
n^{\eta }
T_{n,\mathbf{u},\epsilon} < \infty 
\end{equation}
At first sight, assumption (\ref{e:unifqsconv}) may seem 
strong, but it is  very natural and is satisfied in reasonable situations.
For example suppose we have a contraction scaling limit:
\def\bY{{\mathbf{Y}}}
$\pi _n (dx) \approx h\left({\bx - \bx _0 \over n^{\eta /2}}\right)$. (A special case of this is the Bernstein--von Mises theorem with $\eta =1 $ and $h$ being Gaussian, but our set up is far broader.) If $\{ \bX _t ^n\}$ denotes ScaLE on $\pi _n$, then by simple scaling and time change properties of Brownian motion it is easily checked that if $\bY_t = \bX _{n^{-\eta } t}$ then $\bY $ is (approximately) ScaLE on $h$ which is clearly independent of $n$. Thus to obtain a process which converges in $\mathcal{O} (1)$ we need to {\em slow down} $\bX $ by a time scaling factor of 
\begin{equation}
\label{e:timefactor}
\hbox{time factor }=n^{\eta }\ . \end{equation}
Similar arguments have been used for scaling arguments of other Monte Carlo algorithms that use similar control variates, see for instance the concurrent work of \cite{arxiv:bfr16}.

While posterior contraction has a positive effect on computational cost, it is also the case that for large $n$ the rate at which a likelihood subsample needs to be calculated, as measured by $\tilde{\lambda }$, needs to increase. Since $\tilde{ \lambda }$ depends on the current location in the state space, where we need
to be precise we shall set $\tilde{ \lambda }_{n,K}$ to be an available bound which applies uniformly for $|\bx - \bx _0| <Kn^{-\eta /2}$.

\def\iterC{\mathcal{C}_{\hbox{\small iter}}}
\def\initC{\mathcal{C}_{\hbox{\small init}}}
The following notion of {\em convergence cost} will be required: setting
$$
\iterC = \iterC(n, K,\epsilon ) =  T_{n,K, \epsilon} \cdot  \tilde{\lambda }_{n,K}
$$
ScaLE is said to have {\em iteration complexity} $n^{\varpi}$ or, equivalently, is $\mathcal{O} (n^{\varpi})$ if $
\iterC
(n, K,\epsilon )$ 
is $\mathcal{O} (n^{\varpi})$ for all $K, \epsilon >0$.

Therefore to understand iteration complexity of ScaLE it is necessary to understand the rate at which $\tilde{ \lambda }_{n,K}$ grows with $n$.
A general way to do this is to use global, or local, bounds on the second-derivatives of the log-likelihood for each datum. To simplify the following exposition  a global bound is assumed, so that  
\begin{align}
 \rho (\nabla^2\log f_I(\mathbf{x}))\leq P_{n}, \label{eq:hessiancon}
\end{align}
for some $P_{n}>0$, where $\rho (\cdot)$ represents the spectral radius and $\nabla ^2$ is the Hessian matrix.  For smooth densities with Gaussian and heavier tails, the Hessian of the log-likelihood is typically uniformly bounded (in both data and parameter). In such cases such a global bound would be expected, and in fact  {$P_{n}$} would be constant in $n$.\\
\\
Recalling the layer construction of \secref{s:langevin} {for a single trajectory of killed Brownian motion,}  we can ensure that over any finite time interval we have $\mathbf{x} \in \mathcal{H}$, some hypercube. Let the centre of the hypercube be $\mathbf{x}^*$.\\
\\
In this section, eventually the assumption that the posterior contracts at a rate $n^{-{\eta /2}}$ will be made, i.e. that
$\{n^{\eta /2}(\mathbf{x} - \mathbf{x}_0), n=1, 2, \ldots \}$ is tight.
 The so-called {\em regular case} corresponds to the case where $\eta =1$, although there is no need to make any explicit assumptions about normality in the following. 
 The practitioner has complete freedom to choose $\mathcal{H}$, and it makes sense to choose this so that $\|\mathbf{x}-\mathbf{x}^*\|< C^*n^{-\eta/2}$ for some $C^*>0$ and for all $\mathbf{x} \in \mathcal{H}$.\\
\\
It is possible to bound $\tilde{\phi}_A(\mathbf{x})$ both above and below if we can bound  $|\tilde{\phi}_A(\mathbf{x})|$  over $\mathcal{A}$-almost all possible realisations of {$A$}. To bound $|\tilde{\phi}_A(\mathbf{x})|$, the approach here is to  first consider the elementary estimator in (\ref{eq:alphatilde}). By imposing the condition in (\ref{eq:hessiancon}) we can then obtain
\begin{align}
\max_{\mathbf{x} \in \mathcal{H},I\in\{0,\ldots{},n\}}|\tilde{\alpha}_I(\mathbf{x})| \leq (n+1)\cdot P_{{n}} \cdot \max_{\mathbf{x} \in \mathcal{H}}\|\mathbf{x}-\hat{\mathbf{x}}\|.
\end{align}
Thus it is possible to bound the estimator in (\ref{eq:subsamplelambda}) as follows 
\begin{align}
& 2\max_{\mathbf{x} \in \mathcal{H},A\in\mathcal{A}}|\tilde{\phi}_A(\mathbf{x})-  C
| \leq \nonumber \\
&  
{(n+1)P_n 
\max_{\mathbf{x} \in \mathcal{H}} \|\mathbf{x}-\hat{\mathbf{x}}\| 
\left[
|2\nabla\log\pi(\hat{\mathbf{x}})| +P_{{n}} 
(n+1) \max_{\mathbf{x} \in \mathcal{H}} \|\mathbf{x}-\hat{\mathbf{x}}\|
\right]
 +P_{{n}} d  (n+1)
  }.
\end{align}
We can use the fact that $\max_{\mathbf{x} \in \mathcal{H}}\|\mathbf{x}-\hat{\mathbf{x}}\| \leq \|\mathbf{x}^*-\hat{\mathbf{x}}\|+C^*n^{-\eta /2}$ to bound the terms in this expression.\\
\\ 
We now directly consider the iteration complexity of ScaLE.
We note 
that the appropriate killing rate to ensure mixing in time $\mathcal{O}(1)$ involves slowing down by the time factor given in \ref{e:timefactor},
and is therefore just ${n^{-\eta }}\tilde{\lambda} $.
Assuming $\eta\leq1$, and using the bound on $|\tilde{\phi}_A(\mathbf{x})-C|$ for the hypercube centred on $\mathbf{x}^*$, we have that whilst we remain within the hypercube,
\begin{align}
\label{e:lamtibd}
 \frac{1}{n^{\eta }}\tilde{\lambda} 
 & = \mathcal{O}\Big(
 P_{{n}} n^{1-3\eta /2}
\big(
  P_{{n}} n^{1-\eta /2}
  +|\nabla \log \pi (\hat{\mathbf{x}}) |\big)
  \Big).
  \end{align}
Here the assumption has been made that at stationarity $\mathbf{x}^*$ will be a draw from the support of the posterior, so that under the assumption of posterior contraction at the ${n^{-\eta /2 }}$ rate, then $\|\mathbf{x}^*-\hat{\mathbf{x}}\|=\mathcal{O}_p(n^{-\eta /2})$. This discussion is summarised in the following result.
\begin{theorem} \label{prop:contraction}
Suppose that (\ref{e:unifqsconv}) and (\ref{eq:hessiancon}) hold, posterior contraction occurs at rate $n^{-\eta/2}$ for $\eta\leq 1$, $P_n$ is $\mathcal{O}(1)$ and that $|\nabla \log \pi (\hat{\mathbf{x}}) | = 
\mathcal{O} (n^\iota )$ for some $\iota >0$. Then the iterative complexity of
ScaLE is $\mathcal{O} (n^{\varpi})$ where
$$
\varpi := \max(1-\eta/2,\iota) +1 - 3\eta /2.
$$
In particular, where $\iota \leq 1-\eta/2$, 
$
\varpi = 2- 2\eta.
$
If $\eta = 1$, then it follows that $\varpi =0$ and the iterative complexity of ScaLE is $\mathcal{O}$(1).
\end{theorem}
\noindent This result also illuminates the role played by $|\nabla \log \pi (\hat{\mathbf{x}}) |$ in the efficiency of the algorithm. In the following discussion it is assumed that $\eta =1$. It is worth noting that while a completely arbitrary starting value for  $\hat{\mathbf{x}} $ might make $|\nabla \log \pi (\hat{\mathbf{x}}) |$ an ${\mathcal O}(n)$ quantity leading to an iterative complexity of the algorithm which is ${\mathcal O}(n^{1/2})$. To obtain ${\mathcal O}(1)$ it is simply required that $|\nabla \log \pi (\hat{\mathbf{x}}) |$ be ${\mathcal O}(n^{1/2})$ which gives considerable leeway for any initial explorative algorithm to find a good value for ${\hat{\bx}}$.\\
\\
Note that given bounds on the third derivatives, (\ref{e:lamtibd}) can be improved by linearising the divergence term in (\ref{eq:subsamplelambda}). This idea is exploited later in a logistic regression example (see Sections \ref{s:examp:ulr}, \ref{s:examp:alr}, \ref{s:examp:1}).\\
\\
In the absence of a global bound on the second derivatives, it is possible to replace $P_n$ in the above arguments by any constant that bounds the second-derivatives for all $\mathbf{x}$ such that $\|\mathbf{x}-\hat{\mathbf{x}}\|\leq \max_{\mathbf{x} \in \mathcal{H}} \|\mathbf{x}-\hat{\mathbf{x}}\|$.  In this case, the {\em most extreme} rate at which ${\tilde \lambda}$  can grow is logarithmically with $n$, for instance for light-tailed models where the data really comes from the model being used. Where the tails are mis-specified and light-tailed models are being used, the algorithmic complexity can be considerably worse. There is considerable scope for more detailed analyses of these issues in future work.\\
\\
The above arguments give insight into the impact of our choice of $\hat{\mathbf{x}}$. It affects the bound on $\tilde{\lambda}$, and hence the computational efficiency of ScaLE, through the terms $\|\mathbf{x}^*-\hat{\mathbf{x}}\|$. Furthermore the main term in the order of $\tilde{\lambda}$ is the square of this distance. If $\hat{\mathbf{x}}$ is the posterior mean, then the square of this distance will, on average, be the posterior variance. By comparison, if $\hat{\mathbf{x}}$ is $k$ posterior standard deviations away from the posterior mean, then on average the square distance will be ${k^2+a}$ times the posterior variance (for some constant $a$), and the computational cost of ScaLE will be increased by a factor of roughly $k^2+a$. 

\subsection{Overall complexity}
\label{s:overall}

Here we will briefly discuss the overall complexity of ScaLE.  The general setup of Theorem \ref{prop:contraction} describes the iteration complexity of ScaLE on the assumption that $|\nabla \log \pi (\hat{\mathbf{x}}) |$ grows no worse than $\mathcal{O}(n^\iota )$. However there is a substantial initial computational cost in locating $\hat{\mathbf{x}}$ and calculating $\nabla \log \pi (\hat{\mathbf{x}}) $ which is likely to be $\mathcal{O}(n)$ as there are $n$ terms in the calculation of the latter. Therefore the {\em overall complexity}
of ScaLE can be described as
$$
\mathcal{C} =  \initC + \iterC = \mathcal{O}(n)+ \mathcal{O}(n^{\varpi } t) 
$$
where $t$ represents algorithm time. This is in contrast to an MCMC algorithm for which iteration cost would be $\mathcal{O}(n)$ leading to overall complexity
$tn$. A Laplace approximation will involve an initial cost that is (at very least) $\mathcal{O}(n)$ but no further computation.  
\\
\\
Since they both involve full likelihood calculations, finding the posterior mode and finding  ${\hat x}$ are both likely to be  $\mathcal{O}(n)$ calculations. This can be shown to be the case for strongly log-concave posterior densities \citep{nesterov2013introductory}, though the cost may be higher if the log-posterior is not concave.
On the other hand, the above discussion shows that in order to achieve $\mathcal{O}(1)$
scaling with data we typically only need to find ${\hat x}$ within $\mathcal{O}(n^{-1/2})$ of the posterior model. Thus
finding ${\hat x}$ is certainly no harder than finding the posterior mode, as we can use the same mode-finding algorithm, e.g. \cite{bottou2010large,nesterov2013introductory,jin2017accelera}, but have the option of stopping earlier. 
\\
\\
 If $n$ is sufficiently large that the initialisation cost dominates the iteration cost, ScaLE will computationally be no more expensive to implement than the Laplace approximation. In this case we obtain an {\em exact approximate}  algorithm (ScaLE) for at most the computational complexity of an approximate method (Laplace). These complexity considerations are summarised in Table \ref{t:complex}.
\\
\begin{table}
\caption{Complexity of algorithms for big data. This is split into the complexity of initiation, $\initC$, and the cost of the iterative algorithm, $\iterC$. Here $n$ denotes sample size, $t$ denotes algorithm time, and $\varpi $ and $\eta $ are as given in Theorem \ref{prop:contraction}.
}
\label{t:complex}
\begin{tabular}{| c | c | c | c |}
\hline
&  $\initC$ & $\iterC$ & $\mathcal{C}$ \\
\hline
MCMC &  $0$ & $tn$ &  $tn$  \\
Laplace approximation & $n$ & 0  &  $n$ \\
ScaLE & $n$ & $tn^{\varpi}$ & $n+tn^{\varpi}$ \\
ScaLE when $\eta =1$ & $n$ & $t$ & $n+t$ \\
\hline
\end{tabular}
\end{table}
\\

\section{Theoretical Properties} \label{s:smctheory}
SMC algorithms in both discrete and continuous time have been studied extensively in the literature (for related theory for approximating a fixed-point distribution, including for algorithms with resampling implemented in continuous-time, see \cite{esaim:dm03,DelMoral/Miclo:2000,rousset2006}; a discrete-time algorithm to approximate a fixed-point distribution in a different context was considered by \cite{whiteley2017}). In order to avoid a lengthy technical diversion, we restrict ourselves here to studying a slightly simplified version of the problem in order to obtain the simplest and most interpretable possible form of results. The technical details of this construction are deferred to \apxref{app:smc} and give here only a qualitative description intended to guide intuition and the key result: that the resulting estimator satisfies a Gaussian central limit theorem with the usual Monte Carlo rate.\\
\\
Consider a variant of the algorithm in which (multinomial) resampling occurs at times $kh$ for $k\in \mathbb{N}$ where $h$ is a time step resolution specified in advance and consider the behaviour of estimates obtained at these times. Extension to resampling at a random subset of these resampling times would be possible using the approach of \cite{b:ddj12}, considering precisely the QSMC algorithm presented in \algref{alg:qsmc} and the ScaLE algorithm in \algref{alg:scale} would require additional technical work somewhat beyond the scope of this paper; no substantial difference in behaviour was observed.\\
\\
In order to employ standard results for SMC algorithms it is convenient to consider a discrete time embedding of the algorithms described. We consider an abstract formalism in which between the specified resampling times the trajectory of the Brownian motion is sampled, together with such auxiliary random variables as are required in any particular variant of the algorithm. Provided the potential function employed to weight each particle prior to resampling has conditional expectation (given the path) proportional to the exact killing rate integrated over these discrete time intervals  a valid version of the ScaLE algorithm is recovered.\\
\\
This discrete time formalism allows for results on more standard SMC algorithms to be applied directly to the ScaLE algorithm. We provide  in the following proposition a straightforward corollary to a result in Chapter 9 of \citet{bk:fk}, which demonstrates that estimates obtained from a single algorithmic time slice of the ScaLE algorithm satisfy a central limit theorem. 
\begin{proposition}[Central Limit Theorem]
\label{prop:3}
  In the context described, under mild regularity conditions (see references given in \apxref{app:smc}):
\begin{align*}
\lim_{N \to \infty} \sqrt{N} \left[\frac{1}{N} \sum_{i=1}^N \varphi(X^i_{hk}) -
  \mathbb{E}_{\mathbb{K}_{hk}^x}\left[\varphi(X^i_{hk})\right] \right] \Rightarrow \sigma_k(\varphi) Z
\end{align*}
where, $\varphi:\mathbb{R}^d \to\mathbb{R}$,  $Z$ is a standard normal random variable, $\Rightarrow$ denotes convergence in distribution, and $\sigma_k(\varphi)$ depends upon the precise choice of sub-sampling scheme as well as the test function of interest and is specified in \apxref{app:smc} along with the law $\mathbb{K}_{hk}^x$. 
\end{proposition}

\section{Examples} \label{s:examp}

In this section we present five example applications of the methodology developed in this paper, each highlighting a different aspect of ScaLE and contrasted with appropriate competing algorithms. In particular: in \secref{s:examp:skewed} we consider a simple pedagogical example which has a skewed target distribution, contrasted with MCMC; \secref{s:examp:ulr} considers the performance of a logistic regression model in which significantly less information is available from the data about one of the covariates than the others; in \secref{s:examp:alr} we apply both ScaLE and SGLD to a regression problem based upon the ASA Data Expo Airline On-time Performance data, which is of moderately large size ($\approx 10^8$); \secref{s:examp:1} considers ScaLE applied to a very large logistic regression problem, with a data set of size $n=2^{34}\approx 10^{10.2}$, along with consideration of scalability with respect to data size; Finally, in \secref{s:examp:2} parameter inference for a contaminated regression example is given, motivated by a big data application with $n=2^{27}\approx 10^{8.1}$, and illustrating the potential of an \emph{approximate} implementation of ScaLE even when mis-initialised.\\
\\
All simulations were conducted in R on an Xeon X5660 CPU running at 2.8 GHz. Note that for the purposes of presenting the ScaLE methodology as cleanly as possible, in each example no prior has been specified. In practice, a prior can be simply included within the methodology as described in \secref{s:scale}.\\
\\
Details of data and code used to produce the output within this section can be found in Journal of the Royal Statistical Society: Series B (Statistical Methodology) Datasets.

\subsection{Skewed Target Distribution}\label{s:examp:skewed}

In order to illustrate ScaLE applied to a simple non-Gaussian target distribution, we constructed a small data set of size $n=10$, to which we applied a logistic regression model
\begin{equation}
\label{eq:logistic}
y_i = \left\lbrace \begin{array}{ll} 
1 & \quad\text{with probability } \dfrac{\exp\{\mathbf{x}^T_i \mathbf{\beta}\}}{1+\exp\{\mathbf{x}^T_i \mathbf{\beta}\}},
\\*[10pt] 0 & \quad\text{otherwise}.  \end{array}  \right.
\end{equation}
The data was chosen to induce a skewed target, with $\mathbf{y}^T=(1,1,0,\ldots{},0)$ and $\mathbf{x}^T_{i}=\left(1,(-1)^i/i\right)$.\\
\\
We used the glm R package to obtain the MLE ($\beta^*\approx (-1.5598, -1.3971)$) and observed Fisher information, in order to (mis-)initialise the particles in the ScaLE algorithm. In total $N=2^{10}$ particles were used, along with a sub-sampling mechanism of size $2$ and a control variate computed as in \secref{s:repest} by setting $\hat{\mathbf{x}}=\beta^*$. For comparison we ran random walk Metropolis on the same example initialised at $\beta^*$ using the MCMClogit function provided by MCMCpack \citep{adam_mcmcrun}, computed the posterior marginals based on 1,000,000 iterations thinned to 100,000 and after discarding a burn-in of 10,000 iterations, and overlaid them with together with those estimated by ScaLE in \figref{fig:skewedExample}. These are accompanied by the glm fit used to mis-initialise ScaLE.\\
\\
It is clear from \figref{fig:skewedExample} that the posterior obtained by simulating ScaLE matches that of MCMC, and both identify the skew which would be over-looked by a simple normal approximation. The particle set in ScaLE quickly recovers from its mis-initialisation, and only a modest burn-in period is required. In practice, we would of course not advocate using ScaLE for such a small data setting --- the computational and implementational complexity of ScaLE does not compete with MCMC in this example. However, as indicated in \secref{s:scalecomplexity} and the subsequent examples, ScaLE is robust to increasing data size whereas simple MCMC will scale at best linearly.
\begin{figure}
    \centering
    \includegraphics[width=1\textwidth]{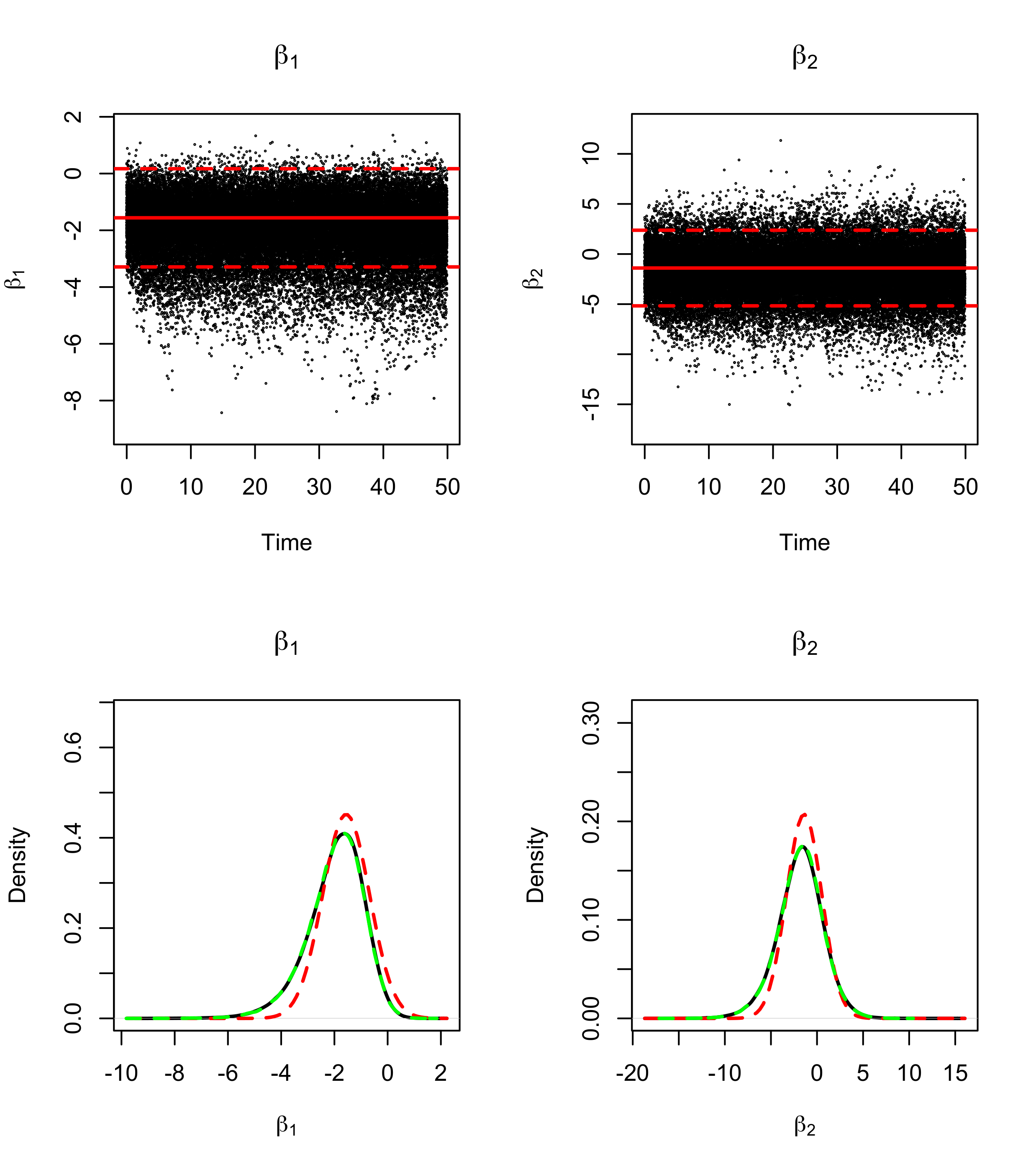}
    \caption{Upper plots are trace trajectories of ScaLE applied to the skewed target distribution example of \secref{s:examp:skewed}. Solid red lines on the trace plots mark the parameter values fitted using the glm R package, and dashed red lines are 95\% confidence intervals imputed using the covariance matrix estimated from the glm package. Lower plots are marginal densities obtained by ScaLE (black lines). Overlaid on the marginal plots are the normal approximation from the glm R package (red dashed line), and from MCMC run (green dashed lines).}
    \label{fig:skewedExample}
\end{figure}

\subsection{Heterogeneous Logistic Regression}\label{s:examp:ulr}
For this example a synthetic data set of size $n=10^7$ was produced from the logistic regression model in (\ref{eq:logistic}). Each record contained three covariates, in addition to an intercept. The covariates were simulated independently from a three-dimensional normal distribution with identity covariance truncated to $[-0.001,0.001] \times [-1,+1]\times[-1,+1]$, and with the true $\mathbf{\beta} = (0,2,-2,2)$ (where the first coordinate corresponds to the intercept). The specification of this data set is such that significantly less information is available from the data about the second covariate than the others. Data was then generated from (\ref{eq:logistic}) using the simulated covariates.\\
\\
As before, the glm R package was used to obtain the MLE and observed Fisher information, which was used within ScaLE to set  $\mathbf{\beta}^\star = \hat{\mathbf{x}}\approx (2.3581\times 10^{-4},\allowbreak 2.3407,\allowbreak -2.0009,\allowbreak 1.9995)$ and  $\precon \approx \textrm{diag}(7.6238\times 10^{-4}, \allowbreak 1.3202,\allowbreak 1.5137\times10^{-3},\allowbreak 1.5138\times10^{-3})$ respectively. For the control variate $\nabla\log\pi(\hat{\mathbf{x}}) \approx (2.0287\times10^{-9},\allowbreak 2.2681 \times 10^{-9},\allowbreak -2.3809 \times 10^{-6},\allowbreak -2.3808 \times 10^{-6})$ was calculated using the full data set, and as expected (and required for computational considerations) is extremely small, along with $\Delta\log\pi(\hat{\mathbf{x}})$.\\
\\
 ScaLE was then applied to this example using $N=2^{10}$ particles initialised using a normal approximation given by the computed $\hat{\mathbf{x}}$ and $\precon$, and a subsampling mechanism of size 2.  The simulation was run for 20 hours, in which time 84,935,484 individual records of the data set were accessed (equivalent to roughly 8.5 full data evaluations). Trace plots for the simulation can be found in \figref{fig:ulr}, along with posterior marginals given by the output (after discarding as burn-in a tenth of the simulation). The posterior marginals are overlaid with the normal approximation given by the R glm fit.\\
 \\
 The estimated means and standard deviations for the regression parameters were ${\bf\mathbf{x}} \approx (-2.3194\times 10^{-4},\allowbreak 2.3197,\allowbreak -2.0009,\allowbreak 1.9995)$, and $\sigma_{\bf x} \approx (7.6703 \times 10^{-4},\allowbreak 1.3296,\allowbreak 1.6386 \times 10^{-4},\allowbreak 1.6217 \times 10^{-4})$ respectively. This is in contrast with $\beta^*$ and  standard deviations of $\approx (7.6238 \times 10^{-4},\allowbreak 1.3203,\allowbreak 1.6237 \times 10^{-4},\allowbreak 1.6233 \times 10^{-4})$ from the glm output.\\
\\
To assess the quality of the output we adopted a standard method for estimating effective sample size (ESS) for a single parameter. In particular, we first estimated a marginal ESS associated with the particles from ScaLE at a single time-point, with this defined as the average of the ratio of the variance of the estimator of the parameter using these particles to the posterior variance of the parameter \citep{iee:ccf99}. To calculate the overall ESS, the dependence of these estimators over-time is accounted for by modelling this dependence as an AR(1) process. Full details of this approach are given in Appendix \ref{app:ESS}. The resulting average ESS per parameter using this approach was found to be 352.\\
\\
The ScaLE output is highly stable and demonstrates that despite the heterogeneity in the information for different parameters, the Bernstein--von Mises limit (Laplace approximation) proves here to be an excellent fit. Although the GLM fit is therefore excellent in this case, ScaLE can be effectively used to verify this. This is in contrast to the example in Section 7.1 where ScaLE demonstrates that the GLM-Laplace approximation is a poor approximation of the posterior distribution.
\begin{sidewaysfigure}
    \centering
    \includegraphics[width=1\textwidth]{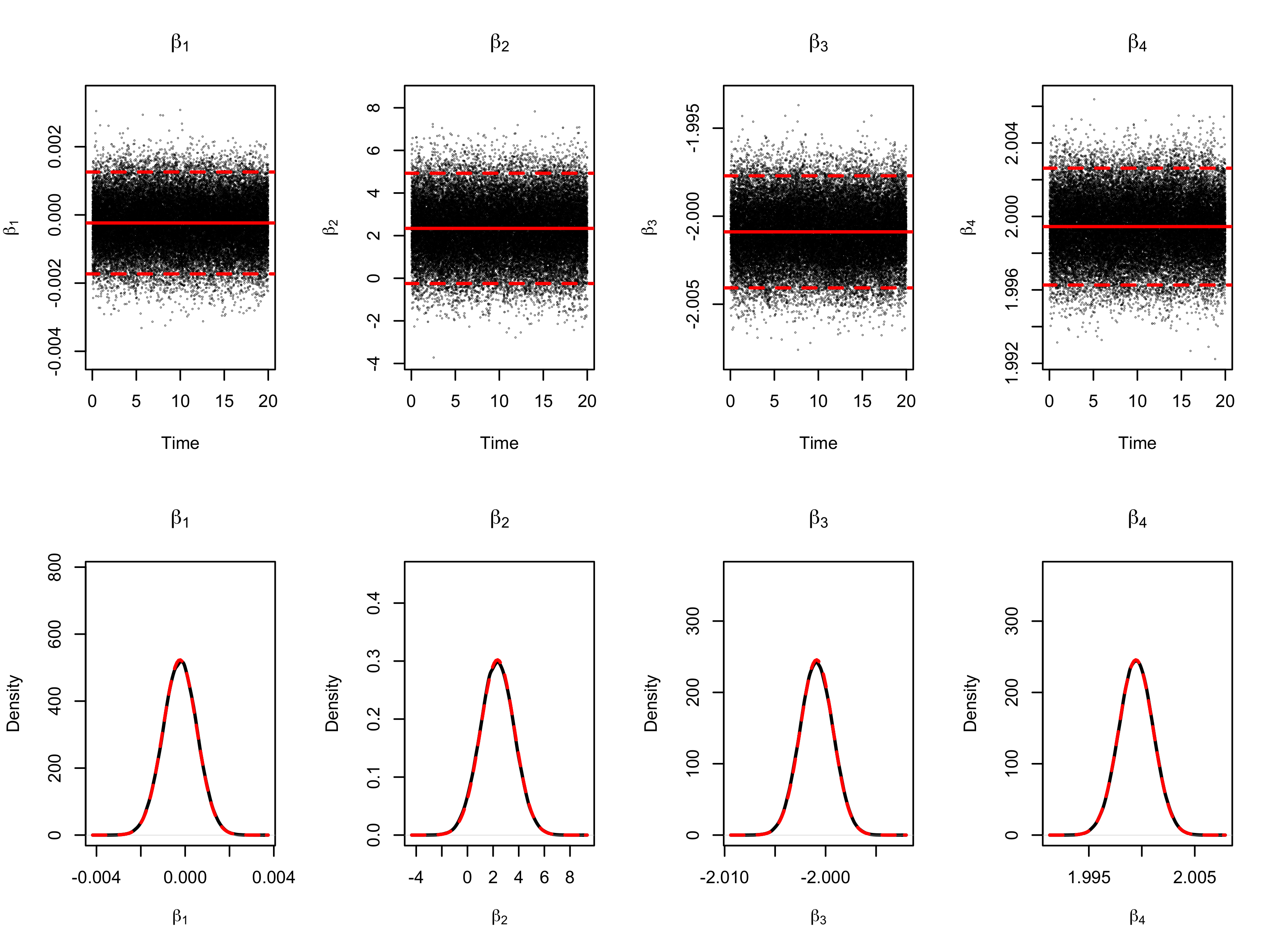}
    \caption{Upper plots are trace plots of ScaLE applied to the heterogeneous logistic regression example of \secref{s:examp:ulr}. Solid red lines on the trace plots mark the parameter values fitted using the glm R package, and dashed red lines are 95\% confidence intervals imputed using the covariance matrix estimated from the glm package. Lower plots are marginal densities obtained by ScaLE (black lines). Overlaid on the marginal plots are the normal approximation from the glm R package (red dotted lines).}
    \label{fig:ulr}
\end{sidewaysfigure}

\subsection{Airline Dataset}\label{s:examp:alr}
To demonstrate our methodology applied to a real (and moderately large) dataset we consider the `Airline on-time performance' dataset which was used for the 2009 American Statistical Association (ASA) Data Expo, and can be obtained from \url{http://stat-computing.org/dataexpo/2009/}. The `Airline' data set consists in its entirety of a record of all flight arrival and departure details for all commercial flights within the USA from October 1987 to April 2008. In total the data set comprises 123,534,969 such flights together with 29 covariates.\\
\\
For the purposes of this example we selected a number of covariates to investigate what effect (if any) they may have on whether a flight is delayed. The Federal Aviation Administration (FAA) considers an arriving flight to be late if it arrives more than 15 minutes later than its scheduled arrival time. As such we take the \textit{flight arrival delay} as our observed data (given by \textbf{ArrDelay} in the Airline data) and treat it as a binary taking a value of one for any flight delayed in excess of the FAA definition.\\
\\
In addition to an intercept, we determine three further covariates which may reasonably affect flight arrival: a \textit{weekend} covariate, which we obtain by treating  \textbf{DayOfWeek} as a binary taking a value of one if the flight operated on a Saturday or Sunday; a \textit{night flight} covariate, which we obtain by taking \textbf{DepTime} (Departure Time) and treating it as a binary taking a value of one if the departure is between 8pm and 5am; and flight \textit{distance}, which we obtain by taking \textbf{Distance} and normalising by subtracting the minimum distance and dividing by the range.\\
\\
The resulting data set obtained by the above process contained a number of missing entries, and so all such flights were omitted from the data set (in total 2,786,730 rows), leaving $n=$120,748,238 rows.  We performed logistic regression taking the \textit{flight arrival delay} variable as the response and treating all other variables as covariates.\\
\\
To allow computation of $\hat{\mathbf{x}}$ and $\precon$ as required by ScaLE the data was first divided into 13 subsets each of size 9,288,326 and the MLE and observed information matrix for each was obtained using the R glm package. It should be noted that the Airline data set is highly structured, and so for robustness the order of the flight records was first permuted before applying glm to the data subsets. An estimate for the MLE and observed information matrix for the full data set was obtained by simply taking the mean for each coefficient of the subset MLE fits, and summing the subset information matrices. The centring point $\hat{\mathbf{x}} \approx (-1.5609,\allowbreak -0.1698,\allowbreak 0.2823,\allowbreak 0.9865)$ was chosen to be the computed MLE fit, and for simplicity $\precon^{-1}$ was chosen to be the square root of the diagonal of the computed information matrix ($\precon \approx \text{diag}(2.309470\times 10^{-4},\allowbreak 4.632830\times 10^{-4},\allowbreak  6.484359\times 10^{-4},\allowbreak 1.2231\times 10^{-5}))$. As before, and as detailed in \secref{s:repest}, we use the full data set in order to compute $\nabla\log\pi(\hat{\mathbf{x}}) \approx (0.00249,\allowbreak 0.0018,\allowbreak 0.0021,\allowbreak 0.0029)$ (which again is small as suggested by the theory, and required for efficient implementation of ScaLE) and $\Delta\log\pi(\hat{\mathbf{x}}) \approx -3.999$. \\
\\
The ScaLE algorithm was initialised using the normal approximation available from the glm fit. In total $N=2^{12}$ particles were used in the simulation, and for the purposes of computing the unbiased estimator $\tilde{\phi}_A(\mathbf{x})$ we used a sub-sample of size $2$. The algorithm was executed so that $n$ individual records of the data set were accessed (i.e. a single access to the full data set), which took 36 hours of computational time. The first tenth of the simulation trajectories were discarded as burn-in, and the remainder used to estimate the posterior density. The trace plots and posterior densities for each marginal for the simulation can be found in \figref{fig:airline_analysis}.\\
\\
For comparison, we also ran stochastic gradient Langevin diffusion (SGLD;  \citet{icml:wt11}). This algorithm approximately simulates from a Langevin diffusion which has the posterior distribution as its stationary distribution. The approximation comes from both simulating an Euler discretised version of the Langevin diffusion and from approximating gradients of the log posterior at each iteration. The approximation error can be controlled by tuning the step-size of the Euler discretisation --- with smaller step-sizes meaning less approximation but slower mixing. We implemented SGLD using a decreasing step-size, as recommended by the theoretical results of \cite{jmlr:ttv16}; and used pilot runs to choose the smallest scale for the step-size schedule which still led to a well-mixing algorithm. As such, the pre-processing expenditure matched that of ScaLE. The accuracy of the estimate of the gradient is also crucial to the performance of SGLD \citep{Dalalyan:2017}, and we employed an estimator that used control variates (similar to that developed in ScaLE) and a mini-batch size of $1000$, following the guidance of \cite{sc:bffn19,arxiv:ndhv17,brosse2018promises}. For comparable results we ensured that SGLD had the same number of log-likelihood evaluations as ScaLE (i.e. equivalent to one single access to the full data set), and initiated SGLD from the centring value used for the control variates. In total the SGLD simulation took 4 hours to execute. The first tenth was discarded as burn-in and the remainder was used to estimate the marginal posteriors, which are overlaid with those estimated by ScaLE in \figref{fig:airline_analysis}.\\
\\
As can be seen in \figref{fig:airline_analysis}, SGLD estimates seem to be unstable here, with the algorithm struggling to mix effectively under the decreasing step size constraint, particularly for the fourth covariate. Indeed, the marginal posteriors should be convex and SGLD deviates strongly from this. This unstable behaviour was confirmed in replicate SGLD runs, and indeed it would be difficult to separate out bias from Monte Carlo error for SGLD without much longer runs. This is in contrast with ScaLE which produces far more stable output in this example.
\begin{sidewaysfigure}
    \centering
    \includegraphics[width=1\textwidth]{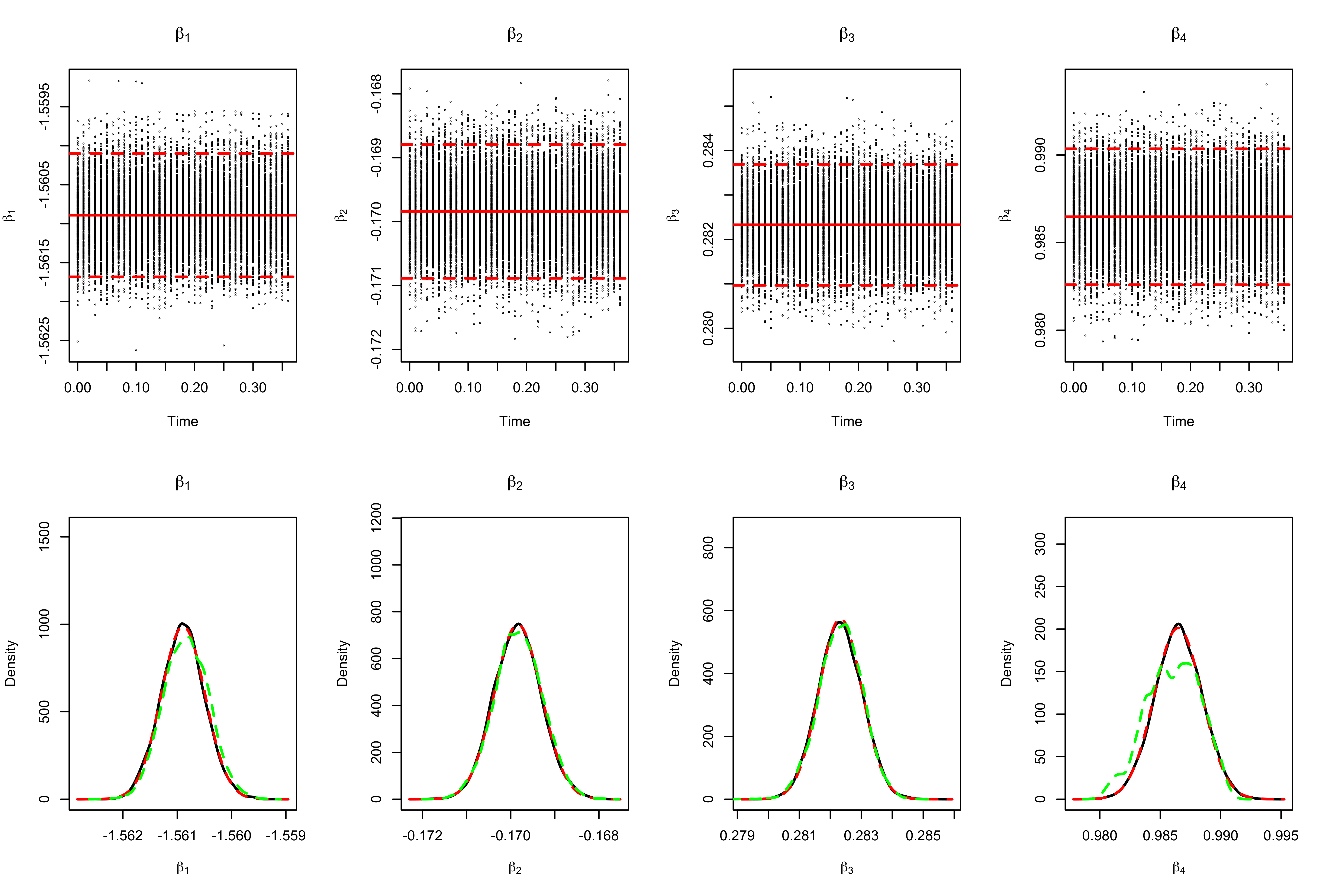}
    \caption{Upper plots are trace plots of ScaLE applied to the Airline data set. Solid red lines on the trace plots mark the parameter values fitted using the glm R package, and dashed red lines are 95\% confidence intervals imputed using the covariance matrix estimated from the glm package. Lower plots are marginal densities obtained by ScaLE (black lines). Overlaid on the marginal plots are the normal approximation from the glm R package (red dotted lines), and from the comparable SGLD run (green dotted lines).}
    \label{fig:airline_analysis}
\end{sidewaysfigure}


\subsection{Large Data Scenario} \label{s:examp:1}

In this subsection we consider an application of ScaLE to a $5$-dimensional logistic regression model, considering data sets of up to size $n=2^{34}\approx 10^{10.2}$. Logistic regression is a model frequently employed within big data settings \citep{ijmsem:sbb16}, and here the scalability of ScaLE is illustrated for this canonical model. In this example, we generate a data set of size $2^{34}$ from this model (\ref{eq:logistic}) by first constructing a design matrix in which the $i^{\text{th}}$ entry ${\bx}_{i} := [1,\zeta_{i,1},\ldots{},\zeta_{i,4}]^T$, where $\zeta_{1,1},\ldots{},\zeta_{n,4}$ are i.i.d. truncated normal random variables with support $[-1,1]$. In the big data setting it is natural to assume such control on the extreme entries of the design matrix, either through construction or physical limitation. Upon simulating the design matrix, binary observations are obtained by simulation using the parameters ${\bf \beta} = [1,1,-1,2,-2]^T$. Due to the extreme size of the data we realised observations only as they were required to avoid storing the entire data set; see code provided for implementation details.\\
\\
First considering the data set of size $n=2^{34}$, then following the approach  outlined in \secref{s:examp:alr}, $\hat{\mathbf{x}}$ and $\precon$ were chosen by breaking the data into a large number of subsets, fitting the R glm package to each subset, then appropriately pooling the fitted MLE and observed Fisher information matrices. In total the full data set was broken into $2^{13}$ subsets of size $2^{21}$, and the glm fitting and pooling was conducted entirely in parallel on a network of $100$ cores. Consequently, $\hat{\mathbf{x}} = \beta^* \approx (0.9999943, \allowbreak 0.9999501,\allowbreak -0.9999813,\allowbreak 1.999987,\allowbreak -1.999982)$ and $\precon \approx \textrm{diag}(1.9710\times 10^{-5},\allowbreak 3.6921\times 10^{-5},\allowbreak 3.6910\times 10^{-5},\allowbreak 3.8339\times 10^{-5},\allowbreak 3.8311\times 10^{-5})$. Upon computing $\hat{\mathbf{x}}$ an additional pass of the $2^{13}$ subsets of the data of size $2^{21}$ was conducted in parallel in order to compute $\nabla\log\pi(\hat{\mathbf{x}}) \approx (-0.0735,\allowbreak -0.0408,\allowbreak 0.0428,\allowbreak -0.09495,\allowbreak 0.0987)$ and $\Delta\log\pi(\hat{\mathbf{x}}) \approx -5.006$ for construction of the control variate. Fully utilising the 100 cores available the  full suite of pre-processing steps required for executing ScaLE  (i.e. the computation of both the glm fit and control variate) took 27 hours of wall-clock time.  \\
\\
ScaLE was applied to this example using $N=2^{10}$ particles initialised using a normal approximation given by the available glm fit, and a subsampling mechanism of size 2.  The simulation was run for 70 hours, in which time 49,665,450 individual records of the data set were accessed (equivalent to roughly 0.0029 full data evaluations). Trace plots for the simulation can be found in \figref{fig:largeexample}. The first tenth of the simulation trajectories were discarded as burn-in, and the remainder used to estimate the posterior density of each marginal. These can also be found in \figref{fig:largeexample}, together with the normal approximation to the posterior marginals given by the R glm fit, is again very accurate here, agreeing closely with the ScaLE output.  Using the ESS approach described in \secref{s:examp:ulr} and Appendix \ref{app:ESS}, the average ESS per parameter was found to be 553.
\begin{sidewaysfigure}
    \centering
    \includegraphics[width=1\textwidth]{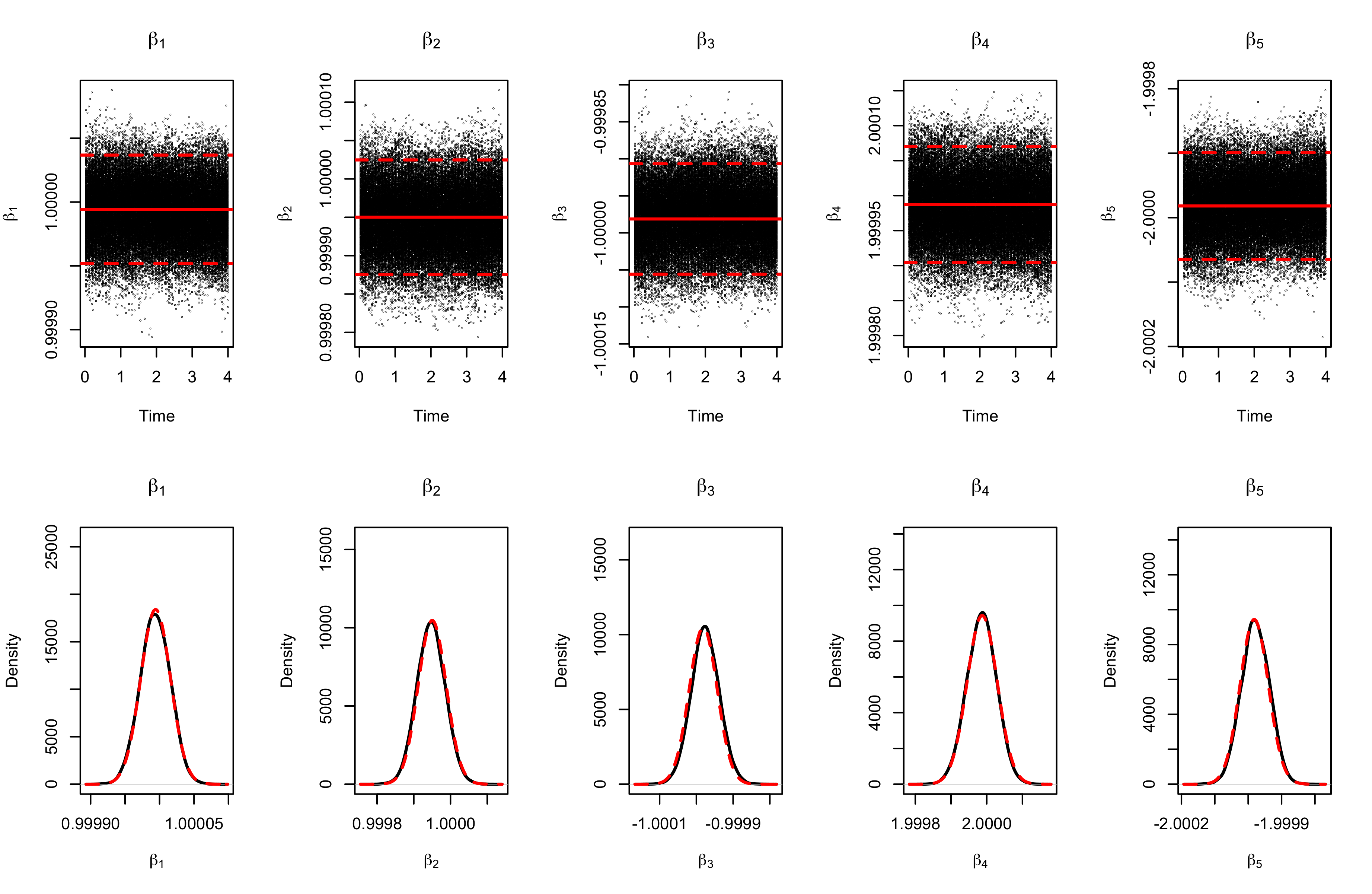}
    \caption{Upper plots are trace plots of ScaLE applied to the large data set of size $2^{34}$ example of \secref{s:examp:1}. Solid red lines on the trace plots mark the parameter values fitted using the glm R package, and dashed red lines are 95\% confidence intervals imputed using the covariance matrix estimated from the glm package. Lower plots are marginal densities obtained by ScaLE (black lines). Overlaid on the marginal plots are the normal approximation from the glm R package (red dotted lines).}
    \label{fig:largeexample}
\end{sidewaysfigure}
\\\\\noindent
To investigate scaling with data size for this example, we considered the same model using the same process as outlined above with data sets varying in size by a factor of 2 from $n=2^{21}$ to $n=2^{33}$.  Computing explicitly the dominating intensity $\tilde{\lambda}_{n,K}$ over the support of the density the relative cost of ScaLE for each data set with respect to the data set of size $n=2^{34}$ can be inferred. This is shown in \figref{fig:largeexample2}.
\begin{figure}
    \centering
    \includegraphics[width=1\textwidth]{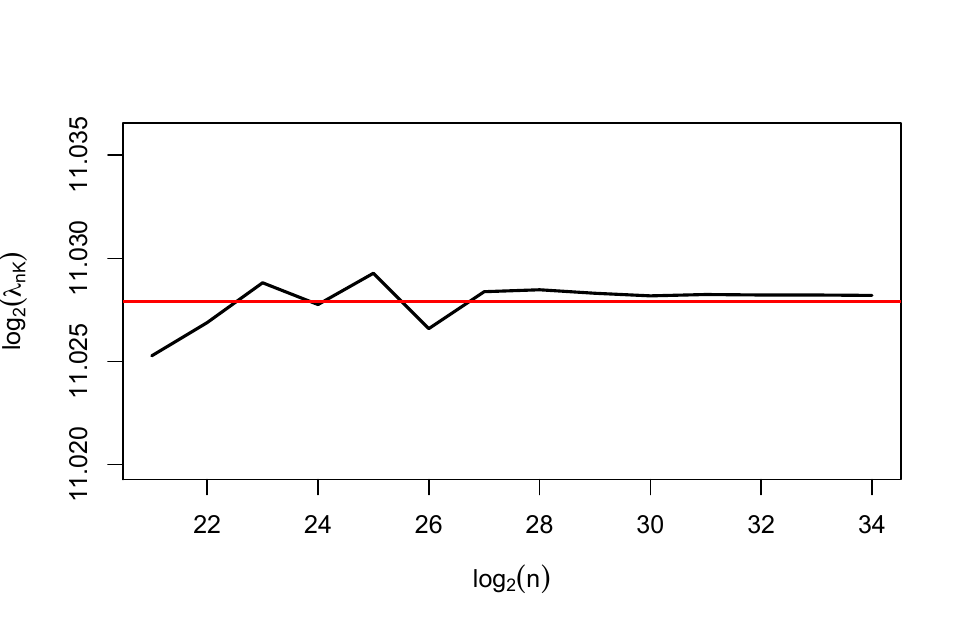}
    \caption{Comparison of the bounding intensities and comparative cost for executing Scale for increasing data set sizes in the large data example of \secref{s:examp:1}.}
    \label{fig:largeexample2}
\end{figure}


\FloatBarrier
\subsection{Contaminated Mixture} \label{s:examp:2}

In this subsection we consider parameter inference for a contaminated mixture model. This is motivated by big data sets obtained from internet applications, in which the large data sets are readily available, but the data is of low quality and corrupted with noisy observations. In particular, in our example each datum comprises two features and a model is fitted in which the likelihood of an individual observation ($y_i$) is,
\begin{align}
F_i := \dfrac{1-p}{\sqrt{2\pi\sigma^2}}\exp\left\{-\dfrac{1}{2\sigma^2}\left(\alpha\cdot x_{i,1} + \beta \cdot x_{i,2} -y_i\right)^2\right\} + \dfrac{p}{\sqrt{2\pi\phi^2}}\exp\left\{-\dfrac{1}{2\phi^2}y_i^2\right\}. \label{eq:contamlike}
\end{align}
In this model $p$ represents the level of corruption and $\phi$ the variance of the corruption. A common approach uses MCMC with data augmentation \citep{jasa:tw87}. However, for large data sets this is not feasible as the dimensionality of the auxiliary variable vector will be $\mathcal{O}(n)$. For convenience a transformation of the likelihood was made so that each transformed parameter is on $\mathbbm{R}$. The details are omitted, and the results presented are given under the original parametrisation.\\
\\
A data set of size $n=2^{27}\approx 10^{8.1}$ was generated from the model with parameters $\mathbf{\mu}=[\alpha,\beta,\sigma,\phi,p]=[2,5,1,10,0.05]$. 
To illustrate a natural future direction for the ScaLE methodology, in this example we instead implemented an approximate version of ScaLE (as opposed to the exact version illustrated in the other examples of \secref{s:examp}). In particular, the primary implementational and computational bottleneck in ScaLE is the formal `localization procedure' to obtain almost sure bounds on the killing rate by constraining Brownian motion to a hypercube (as fully detailed in \secref{s:langevin} and \apxref{s:pathspace layer}). Removing the localization procedure results in the Brownian motion trajectories being unconstrained, and the resulting  dominating intensity $\tilde{\lambda}$ being infinite. However, in practice the probability of such an excursion by Brownian motion outside a suitably chosen hypercube can be made vanishingly small (along with the consequent impact on the Monte Carlo output) by simply adjusting the temporal resolution at which the ergodic average is obtained from the algorithm (noting Brownian motion scaling is $\mathcal{O}(\sqrt{t}$), and inflating the bounds on the Hessian for computing the intensity. The resulting `approximate' algorithm is approximate in a different (more controllable and monitorable) sense than for instance SGLD, but results in substantial (10x-50x) computational speed-ups over the available (but expensive) `exact' ScaLE.\\
\\
In contrast with the other examples of \secref{s:examp}, rather fitting an approximate model in order to initialise the algorithm, instead in this example a single point mass to initialise the algorithm was chosen ($\mu = [2.00045,\allowbreak 5.00025,\allowbreak 0.999875,\allowbreak 10.005\,\allowbreak 0.0499675]$), and this was also used as the point to compute our control variate (described in \secref{s:repest}). The pre-processing for executing ScaLE took approximately $6$ hours of computational time (and is broadly indicative of the length of time a single iteration of an alternative MCMC scheme such as MALA would require). Note that as discussed in \secref{s:scalecomplexity}, this `mis-initialisation' impacts the efficiency of the algorithm by a constant factor, but is however representative of what one in practice may conceivably be able to do (i.e. find by means of an optimisation scheme a point within the support of the target posterior close to some mode, and conduct a single $\mathcal{O}(n)$ calculation). The forgetting of this initialisation is shown in \figref{fig:contam_analysis}.\\
\\
Applying ScaLE for this application we used a particle set of size $N=2^{11}$, and run the algorithm for diffusion time of $T=200$, with observations of each trajectory at a resolution of $t_i-t_{i-1}=0.1$. Again, the choice of $N$ was made as in \secref{s:examp:1} as it provided the required stability. The choice of $T$ was made as it corresponded approximately to a computational budget of one week. \\
\\
Each particle trajectory at each time $t\in[0,T]$ was associated with a sub-sample of the full data set of size $32$, rather than 2, with the resulting likelihood estimates combined by simple averaging. This size was chosen as it provided balance with other components of the algorithm, but allowed stabilisation of the importance weights which was beneficial for the approximate algorithm. In total the entire run required accessing $500$ million individual data points, which corresponds to approximately $4$ full evaluations of the data set.\\
\\
An example of a typical run can be found in \figref{fig:contam_analysis}.  A burn-in period of $100$ was chosen, and alongside the trace plots in \figref{fig:contam_analysis}  an estimate of the marginal density of the parameters is provided using the occupation measure of the trajectories in the interval $t\in[100,200]$.\\
\\
To assess the quality of the simulation, the same batch mean method is employed to estimate the marginal ESS for the run post burn-in as detailed in \secref{s:examp:1}. The mean ESS per dimension for this run was around $930$. An analysis of MALA (for a necessarily much smaller run), indicated it is possible to achieve an ESS of around $T/3$, where $T$ corresponds to the run length subsequent to burn-in. As indicated above, and neglecting burn-in, this would mean an achievable ESS for a comparable computational budget for MALA would be around $10$-$15$.
\begin{sidewaysfigure}
    \centering
    \includegraphics[width=1\textwidth]{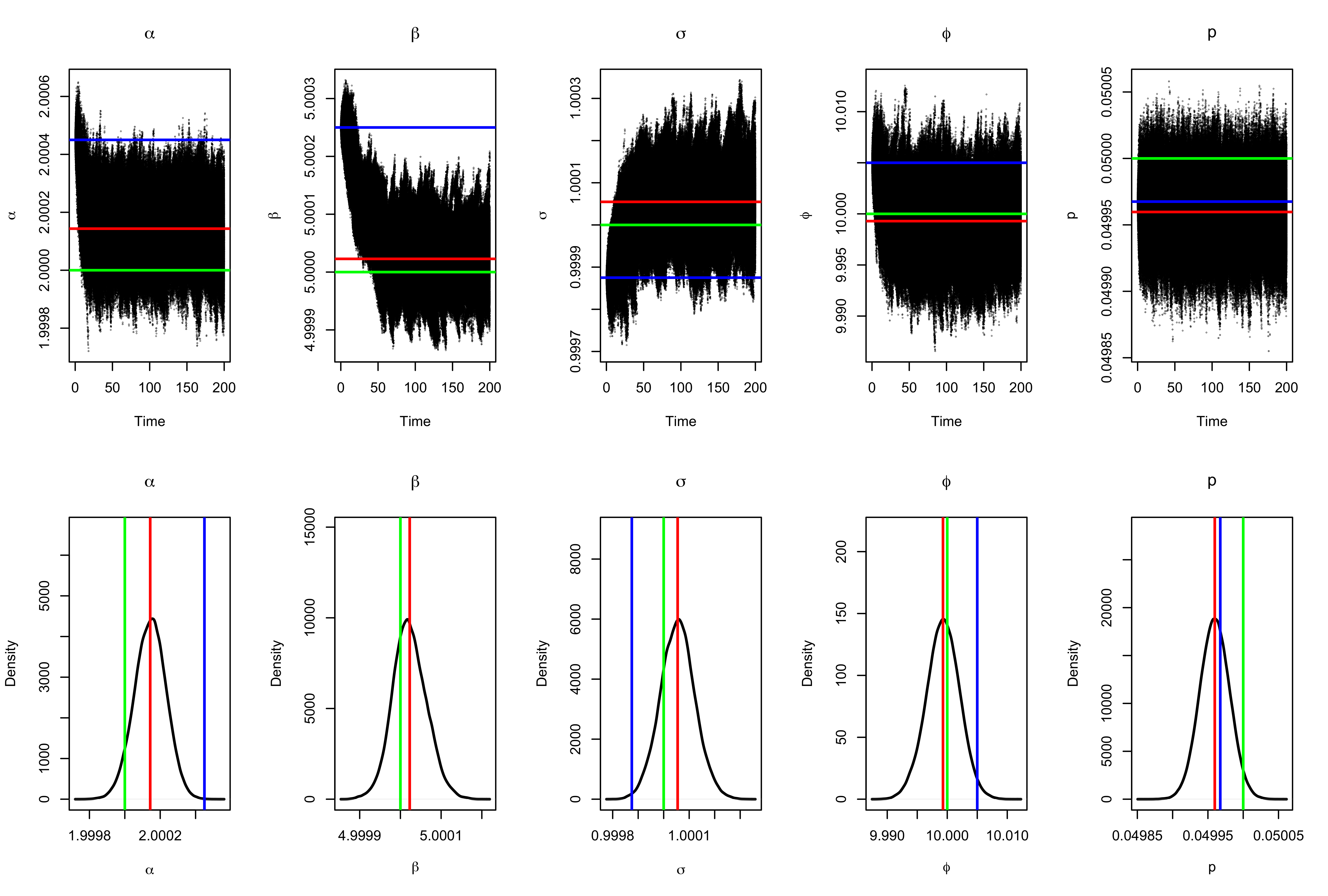}
    \caption{Upper plots are trace trajectories of ScaLE applied to the contaminated mixture data set, lower plots are marginal densities. In both sets of plots, blue lines mark the parameter values used to initialise the algorithm, green lines mark the parameter values the associated data set was generated from, and red lines mark the mean of the marginal densities.}
    \label{fig:contam_analysis}
\end{sidewaysfigure}


\section{Conclusions} \label{s:conclusions}

In this paper we have introduced a new class of \textit{Quasi-Stationary Monte Carlo (QSMC)} methods which are genuinely continuous-time algorithms for simulating from complex target distributions. We have emphasised its particular effectiveness in the context of {\em big data} by developing novel sub-sampling approaches and the \textit{Scalable Langevin Exact (ScaLE)} algorithm. Unlike its immediate competitors, our sub-sampling approach within ScaLE is essentially computationally free and does not result in any approximation to the target distribution. Our methodology is embedded within an SMC framework, supported by underpinning theoretical results. In addition, examples to which ScaLE is applied demonstrate its robust scaling properties for large data sets.\\
\\
We show through theory and examples that computational cost of ScaLE is more stable to data set size than {\em gold standard} MCMC approaches. Moreover we have seen it substantially outperform other approaches such as SGLD which are designed to be robust to data-size at the cost of bias and serial correlation. ScaLE can both confirm that simpler approaches such as Laplace approximation are accurate, and identify when such approximations are poor (as we see in the examples). We see this as a first step in a fruitful new direction for Computational Statistics. Many ideas for variations and extensions to our implementation exist and will stimulate further investigation.\\
\\
Firstly, the need to simulate a quasi-stationary distribution creates particular challenges. Although quasi-stationarity is underpinned by an elegant mathematical theory, the development of numerical methods for quasi-stationarity is understudied.  We have presented an SMC methodology for this problem, but alternatives exist. For instance, \cite{aap:bgz16} suggest alternative approaches.\\
\\
Even within an SMC framework for extracting the quasi-stationary distribution, there are interesting alternatives we have not explored. For example, by a modification of our re-weighting mechanism it is possible to relate the target distribution of interest to the limiting {\em smoothing} distribution of the process, as opposed to the filtering distribution as we do here. Within the quasi-stationary literature this is often termed the type II quasi-stationary distribution. As such, the rich SMC literature offers many other variations on the procedures adopted here.\\
\\
Using SMC  benefits from the rich theory it possesses. However the use of quasi-stationary Monte Carlo actually demands new questions of SMC. \thmref{thm:qsd} gives convergence as $T\to \infty $, while \propref{prop:3} gives a precise description of the limit as the number of particles $N$ increases. There are theoretical and practical questions associated with letting both $N$ and $T$ tend to $\infty $ together.  Within the examples in this paper {\em ad hoc} rules are used to assign computational effort to certain values of $N$ and $T$. However the general question of how to choose these parameters seems completely open. \\
\\
Throughout the paper, we have concentrated on so-called {\em exact approximate} quasi-stationary Monte Carlo methods. Of course in many cases good approximations are good enough and frequently computationally less demanding. However, for many approximate methods it will be difficult to quantify the systematic error being created by the approximation. Moreover, we emphasise that there are different strategies for creating effective approximations that emanate directly from exact approximate methods, and where the approximation error can be well-understood. We have given an example of this in  \ref{s:examp:2}
but other options are possible also.\\
\\
There are interesting options for parallel implementation of SMC algorithms in conjunction with ScaLE. For instance an appealing option would be to  implement the island particle filter \citep{ss:dmov16} which could have substantial effects on the efficiency of our algorithms where large numbers of particles are required. Alternatively one could attempt to embed our scheme in other divide and conquer schemes as described in the introduction.\\
\\
The approach in this paper has concentrated solely on killed (or re-weighted) Brownian motion, and this strategy has been demonstrated to possess robust convergence properties. However, given existing methodology for the exact simulation of diffusions in \cite{aap:br05,b:bpr06,mcap:bpr08,phd:p13,b:pjr15,wsc:p15}, there is scope to develop methods which use proposal measures which much better {\em mimic} the shape of the posterior distribution.\\
\\
The sub-sampling and control variate approaches developed here offer dramatic computational savings for tall data as we see from the examples and from the theory of results like Theorem \ref{prop:contraction}. We have not presented the ScaLE algorithm as a method for high-dimensional inference, and the problem of large $n$ and $d$ will inevitably lead to additional challenges. 
 However there may be scope to extend the ideas of ScaLE still further in this direction. For instance, it might be possible to sub-sample dimensions and thus reduce the dimensional complexity for implementing each iteration.\\
\\
We conclude by noting that as a by-product, the theory behind our methodology offers new insights into problems concerning the existence of quasi-stationary distributions for diffusions killed according to a state-dependent hazard rate, complementing and extending current state-of-the-art literature \citep{tams:se07}.


\section*{Acknowledgments}
This work was supported by the EPSRC (grant numbers EP/D002060/1, EP/K014463/1, EP/N031938/1, EP/R018561/1, EP/R034710/1) and two Alan Turing Institute programmes; the Lloyd's Register Foundation Programme on Data-Centric Engineering; and the UK Government's Defence \& Security Programme.
The authors would like to thank Louis Aslett, Tigran Nagapetyan and Andi Q. Wang for useful discussions on aspects of this paper. In addition, we would like to thank the referees for a number of very helpful suggestions and questions which have improved the exposition of this paper.
\FloatBarrier
\bibliographystyle{imsart-nameyear}
\bibliography{Scale}
\FloatBarrier


\appendix

\FloatBarrier

\section{Proof of \thmref{thm:qsd}} \label{apx:proofthm1}

Here we present a proof of \thmref{thm:qsd}. However, we first formally state the required regularity conditions. We suppose that
\begin{equation}
\label{eq:Q0}
\pi(x)\hbox{ is bounded,}
\end{equation}
and defining $\nu(x)=\pi^2(x)$, we further assume that, for some $\gamma > 0$, 
\begin{align}
\label{eq:Q1}
\liminf_{\mathbf{x} \to \infty } \left(
{\Delta \nu (\bx) \over \nu ^{\gamma + 1/2}(\bx ) } - { \gamma \| \nabla \nu  (\bx)\|^2
\over \nu ^{\gamma + 3/2}(\bx)}
\right) >0, 
\end{align}
where $\Delta $ represents the Laplacian.

\begin{proof}[Proof (Theorem \ref{thm:qsd})] 
Consider the diffusion with generator given by
$$
 \mathfrak{A} f (\bx )= {1\over 2} \Delta f (\bx) +  {1\over 2}  \nabla \log \nu (\bx ) \cdot \nabla f (\bx ).
$$
As $\nu $ is bounded, we assume without loss of generality that its upper bound is $1$.  Our proof shall proceed by checking the conditions of Corollary 6  of \cite{aap:fr05}, which establishes the result. In particular, we need to check that the following are satisfied:
\begin{enumerate}
\item For all $\delta >0$, the discrete time chain $\{X_{n\delta }, n=0, 1, 2,\ldots \}$
is irreducible;
\item
All closed bounded sets are petite;
\item
We can find  a drift function $V(\bx ) =  \nu(\bx )^{-\gamma}$ for some $\gamma >0$, that satisfies the condition
\begin{align}
\label{eq:drift}
{\mathfrak{A}} V^{\eta }(\bx ) \le - c_\eta V(\bx)^{\eta - \alpha} 
\end{align}
for $\bx $ outside some bounded set, for each $\eta \in [\alpha ,1]$ with associated positive constant $c_{\eta }$, and where $\alpha = 1-(2\gamma )^{-1}$.
\end{enumerate}
The first condition holds for any regular diffusion since the diffusion possesses positive continuous transition densities over time intervals $t>0$; and positivity and continuity of the density also implies the second condition. For the final condition we require that
\begin{align}
\label{eq:ratiodrift}
\limsup_{|\bx | \to \infty }
{{\mathfrak{A}} V^{\eta }(\bx )
\over
V^{\eta - \alpha } (\bx)}
<0.
\end{align}
Now by direct calculation
\begin{align}
{\mathfrak{A}} V^{\eta }(\bx ) = {\eta \gamma \over 2}
  \nu  (\bx )^{-\gamma \eta -2} 
\left[
{\eta \gamma \|\nabla \nu (\bx ) \|^2  }
-
\nu (\bx ) \Delta \nu (\bx )
\right],
\end{align}
 so that
\begin{align}
{{\mathfrak{A}} V^{\eta }(\bx )\over V(\bx )^{\eta - \alpha}}
=
{\eta  \gamma \nu  (\bx )^{-3/2-\gamma } \over 2}
\left[
{\eta \gamma \|\nabla \nu (\bx ) \|^2 }
-
\nu (\bx ) \Delta \nu (\bx )
\right].
\end{align}
Therefore (\ref{eq:ratiodrift}) will hold whenever (\ref{eq:Q1}) is true since we have the constraint that $\eta \le 1$ and $\|\nabla \nu (\bx ) \|^2$ is clearly non-negative. As such the result holds as required.
\end{proof}
\noindent 
Note that the condition in (\ref{eq:Q1}) is essentially a condition on the tail of $\nu$. This will hold even for heavy-tailed distributions, and we show this is the case for a class of 1-dimension target densities in \apxref{apx:polynomial}.


\section{Polynomial tails} \label{apx:polynomial}

In this appendix we examine condition (\ref{eq:Q1}) which we use within \thmref{thm:qsd}. This is essentially a condition on the tail of $\nu$, and so we examine the critical case in which the tails of $\nu$ are heavy. More precisely, we demonstrate that for polynomial tailed densities in one-dimension that (\ref{eq:Q1}) essentially amounts to requiring that $\nu^{1/2}$ is integrable. Recall that by construction $\nu^{1/2}$ will be integrable as  
we have chosen $\nu^{1/2}=\pi$.\\
\\
For simplicity, suppose that $\nu $ is a density on $[1,\infty )$ such that $\nu (x) = x^{-p}$. In this case we can easily compute that for $p>1$,
\begin{align}
\nabla \nu (x) & = -px^{-p-1} \nonumber\\
 \Delta \nu (x) & = p(p+1) x^{-p-2} \nonumber
\end{align}
from which we can easily compute the quantity whose limit is taken in (\ref{eq:Q1}) as
\begin{align}
x^{p(\gamma - 1/2) -2} [p(p+1) - \gamma p^2]. \nonumber
\end{align}
As such, we have that condition (\ref{eq:Q1}) holds if and only if
\begin{equation}
\label{eq:pol1}
p+1 - \gamma p >0\end{equation}
and 
\begin{align}
\label{eq:pol2}
p(\gamma - 1/2) - 2 \ge 0.
\end{align}
\noindent Now we shall demonstrate that we can find such $\gamma $ for all $p>2$. For instance, suppose that $p=2 + \epsilon $. The case $\epsilon \ge 2$ can be handled by just setting $\gamma =1$, so suppose otherwise and set $\gamma = 3/2-\epsilon/4$. In this case, (\ref{eq:pol2}) just gives $\epsilon/2 - \epsilon^2/4\ge 0$. Moreover the expression in (\ref{eq:pol1}) becomes $3\epsilon/2 + \epsilon^2>0$, completing our argument.


\section{Simulation of a Path-Space Layer and Intermediate Points} \label{s:pathspace layer} \label{apx:fpttheory}

In this appendix we present the methodology and algorithms required for simulating an individual proposal trajectory of (layered) killed multivariate Brownian motion, which is what is required in \secref{s:emcfd}. Our exposition is as follows: In \apxref{ss:fpt} we present the work of \cite{spl:d09}, in which a highly efficient rejection sampler is developed (based on the earlier work of \cite{macis:bj08}) for simulating the first passage time for univariate standard Brownian motion for a given symmetric boundary, extending it to consider the case of the univariate first passage times of $d$-dimensional standard Brownian motion with non-symmetric boundaries. This construction allows us to determine an interval (given by the first, first passage time) and layer (a hypercube inscribed by the user specified univariate boundaries) in which the sample path is almost-surely constrained, and by application of the strong Markov property can be applied iteratively to find, for any interval of time, a layer (a concatenation of hypercubes) which almost-surely constrains the sample path; In \apxref{ss:fptinter} we present a rejection sampler enabling the simulation of constrained univariate standard Brownian motion as developed in \secref{ss:fpt}, at any desired intermediate point. As motivated in \secref{s:emcfd} these intermediate points may be at some random time (corresponding to a proposed killing point of the proposed sample path), or a deterministic time (in which the sample path is extracted for inclusion within the desired Monte Carlo estimator of QSMC (\ref{eq:occupation2})); Finally, in \apxref{ss:fptfinal} we present the full methodology required in Sections \ref{s:emcfd} and \ref{s:scale} in which we simulate multivariate Brownian motion at any desired time marginal, with $d$-dimensional hypercubes inscribing intervals of the state space in which the sample path almost surely lies.


\subsection{Simulating the first passage times of univariate and multivariate standard Brownian motion} \label{ss:fpt}

To begin with we restrict our attention to the ($i^{\text{th}}$) dimension of multivariate standard Brownian motion initialised at $0$, and the first passage time of the level $\theta^{(i)}$ (which is specified by the user). In particular we denote,
\begin{align}
\tau^{(i)} := \inf\{t\in\mathbbm{R}+\,:\,|W^{(i)}_t-W^{(i)}_0| \geq \theta^{(i)}\}.
\end{align}
Recalling the self similarity properties of Brownian motion (\cite[Section 2.9]{bk:bmsc}), we can further restrict our attention to the simulation of the first passage time of univariate Brownian motion of the level $1$, noting that $\tau^{(i)}\overset{\mathcal{D}}{=}\left(\theta^{(i)}\right)^2\bar{\tau}$ where,
\begin{align}
\bar{\tau} := \inf\{t\in\mathbbm{R}+\,:\,|W_t-W_0| \geq 1\},
\end{align}
noting that at this level,
\begin{align}
\mathbbm{P}(W_\tau=W_0+1)=\mathbbm{P}(W_\tau=W_0-1)=\dfrac{1}{2}.
\end{align}
Denoting by $f_{\bar{\tau}}$ the density of $\bar{\tau}$ (which cannot be evaluated point-wise), the approach outlined in \cite{spl:d09} for drawing random samples from $f_{\bar{\tau}}$ is a series sampler. In particular, an accessible dominating density of $f_{\bar{\tau}}$ is found (denoted $g_{\bar{\tau}}$) from which exact proposals can be made, then upper and lower monotonically convergent bounding functions are constructed ($\lim_{n\to\infty} f^{\uparrow}_{\bar{\tau},n} \to f_{\bar{\tau}}$ and $\lim_{n\to\infty} f^{\downarrow}_{\bar{\tau},n} \to f_{\bar{\tau}}$ such that for any $t\in\mathbbm{R}_+$ and $\epsilon>0$ $\exists\,n^*(t,\epsilon)$ such that $\forall\,n \geq n^*(t,\epsilon)$ we have $f^{\uparrow}_{\bar{\tau},n}(t) - f^{\downarrow}_{\bar{\tau},n}(t) < \epsilon$), and then evaluated to sufficient precision such that acceptance or rejection can be made while retaining exactness. A minor complication arises in that no known, tractable dominating density is uniformly efficient on $\mathbbm{R}_+$, and furthermore no single representation of the bounding function converge monotonically to the target density point-wise on $\mathbbm{R}_+$. As such, the strategy deployed by \cite{spl:d09} is to exploit a dual representation of $f_{\bar{\tau}}$ given by \cite{ams:ct62} in order to construct a hybrid series sampler, using one representation of $f_{\bar{\tau}}$ for the construction of a series sampler on the interval $(0,t_1]$ and the other representation for the interval $[t_2,\infty)$ (fortunately we have $t_1>t_2$, and so we have freedom to choose a threshold $t^*\in[t_2,t_1]$ in which to splice the series samplers). In particular, as shown in \cite{ams:ct62} $f_{\bar{\tau}} (t) =  \pi\sum_{k=0}^{\infty} (-1)^k a_k(t)$ where, the elements of the two expansions are given by
\begin{equation}
a_k(t) = \left\lbrace \begin{array}{ll} 
\left(\dfrac{2}{\pi t} \right)^{{3}/{2}} \left(k+\frac{1}{2} \right) \exp \left\{-\dfrac{2}{t}(k+\frac{1}{2})^2\right\}, & \quad\quad(1)
\\*[10pt] \left(k+\frac{1}{2} \right)\exp \left\{-\dfrac{1}{2}(k+\frac{1}{2})^2 \pi^2 t\right\}, & \quad\quad (2) \label{eq:alt1and2} \end{array}  \right.
\end{equation}
and so by consequence upper and lower bounding sequences can be constructed by simply taking either representation and truncating the infinite sum to have an odd or even number of terms respectively (and thresholding to between zero and the proposal, $g_{\bar{\tau}}$, introduced below). More precisely, 
\begin{align}
f^{\downarrow}_{\bar{\tau},n}(t)
:=\left(\pi\sum_{k=0}^{2n+1} (-1)^k a_k(t)\right)_+,
& \quad \text{}\quad f^{\uparrow}_{\bar{\tau},n}(t)
:=\left[\pi\sum_{k=0}^{2n} (-1)^k a_k(t)\right]\wedge g_{\bar{\tau}}(t).
\end{align}
As shown in Lemma 1 of \citet{spl:d09}, the bounding sequences based on the representation of $f_{\bar{\tau}} (t)$ in (\ref{eq:alt1and2}.1) are monotonically converging for $t \in (0,4/\log(3)]$, and for (\ref{eq:alt1and2}.2) monotonically converging for $t \in [\log(3)/\pi^2,\infty)$. After choosing a suitable threshold $t^*\in[4/\log(3),\log(3)/\pi^2]$ for which to splice the series samplers, then by simply taking the first term in each representation of $f_{\bar{\tau}} (t)$ a dominating density can be constructed as follows,
\begin{align}
f_{\bar{\tau}} (t) \leq g_{\bar{\tau}} (t) \propto\underbrace{\dfrac{2}{\pi t^{3/2}}\exp \left\{-\dfrac{1}{2t} \right\}\cdot \mathbbm{1}_{t \leq t^*}}_{\propto g^{(1)}_{\bar{\tau}}(t)} + \underbrace{\dfrac{\pi}{2} \exp \left\{-\dfrac{\pi^2 t}{8} \right\} \cdot \mathbbm{1}_{t \geq t^*}}_{\propto g^{(2)}_{\bar{\tau}}(t)}. \label{eq:gtau}
\end{align}
\cite{spl:d09} empirically optimises the choice of $t^*=0.64$ so as to minimise the normalising constant of (\ref{eq:gtau}). With this choice $M_1:=\int  g^{(1)}_{\bar{\tau}}(t) \ud t \approx 0.422599$ (to 6 d.p.) and $M_2:=\int  g^{(2)}_{\bar{\tau}}(t) \ud t \approx 0.578103$ (to 6 d.p.), and so we have a normalising constant $M=M_1+M_2\approx 1.000702$ (to 6 d.p.) which equates to the expected number of proposal random samples drawn from $g_{\bar{\tau}}$ before one would expect an accepted draw (the algorithmic `outer loop'). Now considering the iterative algorithmic `inner loop' -- in which the bounding sequences are evaluated to precision sufficient to determine acceptance or rejection -- as shown in \cite{spl:d09}, the exponential convergence of the sequences ensures that the number of iterations required is uniformly bounded in expectation by $3$.\\
\\
Simulation from $g_{\bar{\tau}}$ is possible by either simulating $\bar{\tau}\sim g^{(1)}_{\bar{\tau}}$ with probability $M_1/M$, else $\bar{\tau}\sim g^{(2)}_{\bar{\tau}}$. Simulating $\bar{\tau}\sim g^{(1)}_{\bar{\tau}}$ can be achieved by noting that, for $t \sim g_{\bar\tau}^1$, $t\overset{\mathcal{D}}{=}t^*+8X/\pi^2$, where $X\sim\Exp(1)$. Simulating $\bar{\tau}\sim g^{(2)}_{\bar{\tau}}$ can be achieved by noting that as outlined in \cite[IX.1.2]{bk:nurvg}, for $t \sim g_{\bar\tau}^2$, $t\overset{\mathcal{D}}{=}t^*/(1+t^*X)^2$, where $X:=\inf_{i}\{\left\{X_i\right\}^\infty_{i=1}\overset{\text{iid}}{\sim}\Exp(1):(X_i)^2\leq 2X_{i+1}/t^*,(i-1)/2\in\mathbbm{Z}\}$.\\
\\
A summary of the above for simulating jointly the first passage time and location of the $i^{\text{th}}$ dimension of Brownian motion of the threshold level $\theta^{(i)}$ is provided in \algref{alg:devroyefpg}.\\
\\
\begin{algorithm}[h]
	\caption{Simulating $(\tau,W^{(i)}_{\tau})$, where $\tau:= \inf\{t\in\mathbbm{R}+:|W^{(i)}_t-W^{(i)}_0| \geq \theta^{(i)}\}$ \citep{spl:d09}.} \label{alg:devroyefpg}
	\begin{enumerate}
	\item Input $W^{(i)}_0$ and $\theta^{(i)}$.
	\item $g_{\bar{\tau}}$: Simulate $u\sim\U[0,1]$, \label{alg:st:proposal}
	\begin{enumerate}
	\item $g^{(1)}_{\bar{\tau}}$: If $u\leq M_1/M$, then simulate $X\sim\Exp(1)$ and set $\bar{\tau}:= t^*+8X/\pi^2$.
		\item $g^{(2)}_{\bar{\tau}}$: If $u> M_1/M$, then set $X:=\inf_{i}\{\left\{X_i\right\}^\infty_{i=1}\overset{\text{iid}}{\sim}\Exp(1):(X_i)^2\leq 2X_{i+1}/t^*,(i-1)/2\in\mathbbm{Z}\}$ and set $\bar{\tau}:=t^*/(1+t^*X)^2$.
	\end{enumerate}
	\item $u$: Simulate $u\sim\U[0,1]$ and set $n=0$.
	\item $f^{\cdot}_{\bar{\tau},n}$: While $u\cdot g_{\bar{\tau}}(\bar{\tau})\in(f^{\downarrow}_{\bar{\tau},n}(\bar{\tau}),f^{\uparrow}_{\bar{\tau},n}(\bar{\tau}))$, set $n=n+1$.
	\item $f_{\bar{\tau}}$: If $u\cdot g_{\bar{\tau}}(\bar{\tau})\leq f^{\downarrow}_{\bar{\tau},n}(\bar{\tau})$ accept, else reject and return to \stepref{alg:st:proposal}.
	\item $\tau$: Set $\tau:=(\theta^{(i)})^2\bar{\tau}$.
    \item $W^{(i)}_\tau$: With probability $1/2$ set $W^{(i)}_\tau=W^{(i)}_0+\theta^{(i)}$, else set $W_\tau^{(i)}=W^{(i)}_0-\theta^{(i)}$.
    \item Return $(\tau,W^{(i)}_\tau)$.
	\end{enumerate}
\end{algorithm}
\noindent \!\!\!Note that generalising to the case where we are interested in the first passage time of Brownian motion of a non-symmetric barrier, in particular for $\ell^{(i)},\upsilon^{(i)}\in\mathbbm{R}_+$,
\begin{align}
\tau^{(i)} := \inf\{t\in\mathbbm{R}+\,:\,W^{(i)}_t-W^{(i)}_0 \notin (W^{(i)}_0-\ell^{(i)},W^{(i)}_0+\upsilon^{(i)})\},
\end{align}
is trivial algorithmically. In particular, using the strong Markov property we can iteratively apply \algref{alg:devroyefpg} setting $\theta^{(i)}:=\min(\ell^{(i)},\upsilon^{(i)})$ and simulating intermediate first passage times of lesser barriers, halting whenever the desired barrier is attained. We suppress this (desirable) flexibility in the remainder of the paper to avoid the resulting notational complexity.


\subsection{Simulating intermediate points of multivariate standard Brownian motion conditioned on univariate first passage times} \label{ss:fptinter}

Clearly in addition to being able to simulate the first passage times of a single dimension of Brownian motion, we want to be able simulate the remainder of the dimensions of Brownian motion at that time, or indeed the sample path at times other than its first passage times. As the dimensions of standard Brownian motion are independent (and so Brownian motion can be composed by considering each dimension separately), we can restrict our attention to simulating a single dimension of the sample path for an intermediate time $q\in[s,\tau]$  given $W_s$, the extremal value $W_\tau$, and constrained such that $\forall u\in[s,\tau], W_u\in[W_s-\theta,W_s+\theta]$. Furthermore, as we are interested in only the forward simulation of Brownian motion, by application of the strong Markov property we need consider only the simulation of a single intermediate point (although by application of \cite[Section 7]{b:pjr15} simulation at times conditional on future information is possible).\\
\\
To proceed, note that (as outlined in \cite[Prop. 2]{aap:agp95}) the law of a univariate Brownian motion sample path in the interval $[s,\tau]$ (where $s<\tau$) initialised at $(s,W_s)$ and constrained to attain its extremal value at $(\tau,W_\tau)$, is simply the law of a three-dimensional Bessel bridge. We require the additional constraint that $\forall u\in[s,\tau], W_u\in[W_s-\theta,W_s+\theta]$, which can be imposed in simulation by deploying a rejection sampling scheme in which a Bessel bridge sample path is simulated at a single required point (as above) and accepted if it meets the imposed constraint at either side of the simulated point, and rejected otherwise.\\
\\
As presented in \cite{b:bpr06,phd:p13}, the law of a Bessel bridge sample path (parametrised as above) coincides with that of an appropriate time rescaling of three independent Brownian bridge sample paths of unit length conditioned to start and end at the origin (denoted by $\{b^{(i)}\}^3_{i=1}$). Supposing we require the realisation of a Bessel bridge sample path at some time $q\in[s,\tau]$, then by simply realising three independent Brownian bridge sample paths at that time marginal ($\{b^{(i)}_q\}^3_{i=1}$), we have,
\begin{align}
W_q & = W_s + (-1)^{\mathbbm{1}(W_\tau<W_s)} \sqrt{(\tau-s)\left[\left(\dfrac{\theta(\tau-q)}{(\tau-s)^{3/2}} + b^{(1)}_q\right)^2 + (b^{(2)}_q)^2 +(b^{(3)}_q)^2\right]}.
\end{align}
The method by which the proposed Bessel bridge intermediate point is accepted or rejected (recall, to impose the constraint that $\forall u\in[s,\tau], W_u\in[W_s-\theta,W_s+\theta]$) is non-trivial as there does not exist a closed form representation of the required probability (which we will denote in this appendix by $p$). Instead, as shown in \thmref{thm:besselaccprob}, a representation for $p$ can be found as the product of two infinite series, which as a consequence of this form cannot be evaluated directly in order to make the typical acceptance-rejection comparison (i.e. determining whether $u\leq p$ or $u>p$, where $u\sim\U[0,1]$). The strategy we deploy to retain exactness and accept with the correct probability $p$ is that of a retrospective Bernoulli sampler \cite[Sec. 6.0]{b:pjr15}. In particular, in \corrolref{cor:besselacrej} we construct monotonically convergent upper and lower bounding probabilities ($p^{\uparrow}_n$ and $ p^{\downarrow}_n$ respectively) with the property that $\lim_{n\to\infty} p^{\uparrow}_n \to p$ and $\lim_{n\to\infty} p^{\downarrow}_n \to p$ such that for any $u\in[0,1]$ and $\epsilon>0$ $\exists\,n^*(t)$ such that $\forall\,n \geq n^*(t)$ we have $p^{\uparrow}_{n} - p^{\downarrow}_{n} < \epsilon$, which are then evaluated to sufficient precision to make the acceptance-rejection decision, taking almost surely finite computational time.
\FloatBarrier
\begin{theorem}  \label{thm:besselaccprob}
The probability that a three-dimensional Bessel bridge sample path $W\sim \left.\mathbbm{W}^{W_s,W_\tau}_{s,\tau}\,\mvbar\,(W_\tau, W_q)\right.$ for $s<q<\tau$ attaining its boundary value at $(\tau,W_\tau)$, remains in the interval $[W_s-\theta,W_s+\theta]$, can be represented by the following product of infinite series (where we denote by $m:=\mathbbm{1}(W_\tau>W_s)-\mathbbm{1}(W_\tau<W_s)$),
\begin{align}
& \mathbbm{P}\left(W_{[s,\tau]}\in [W_s\!-\!\theta,W_s\!+\!\theta]|W_s,W_q,W_\tau\right)
 \nonumber\\
& \quad\quad\quad = \underbrace{\left(\dfrac{1 -\sum^{\infty}_{j=1}\big[\varsigma_{q-s}(j;W_s-W_q,\theta)-\varphi_{q-s}(j;W_s-W_q,\theta)\big]}{1-\exp\left\{-2\theta[m(W_s-W_q)+\theta]/(q-s)\right\}}  \right)}_{=:p_1} \nonumber\\
& \quad\quad\quad\quad\quad\quad \cdot \underbrace{\left(1 + \sum^{\infty}_{j=1}\big[ \psi_{\tau-q}(j;W_q-W_\tau,\theta,m) + \chi_{\tau-q}(j;W_q-W_\tau,\theta,m)\big]  \right)}_{=:p_2}, 
\end{align}
where,
\begin{align}
\varsigma_{\Delta}(j;\delta,\theta)
& := 2\cdot\exp\left\{-\dfrac{2\theta^2(2j-1)^2}{\Delta}\right\}\cdot\cosh\left(\dfrac{2(2j-1)\theta \delta}{\Delta}\right),
\end{align}
\begin{align}
\varphi_{\Delta}(j;\delta,\theta)
& := 2\cdot\exp\left\{-\dfrac{8\theta^2j^2}{\Delta}\right\}\cdot\cosh\left\{\dfrac{4\theta\delta j}{\Delta}\right\},
\end{align}
\begin{align}
\psi_{\Delta}(j;\delta,\theta,m) := \chi_{\Delta}(j;\delta,\theta,-m) 
& := \dfrac{(4\theta j +m\delta)}{m\delta}\cdot\exp\left\{-\dfrac{4\theta j}{\Delta}(2\theta j +m\delta)\right\}.
\end{align}
\proof Begin by noting that that the strong Markov property allows us to decompose our required probability as follows,
\begin{align}
&\mathbbm{P}\left(W_{[s,\tau]}\in[W_s-\theta,W_s+\theta]|W_s,W_q,W_\tau\right) \nonumber\\
&\,\, = \underbrace{\mathbbm{P}\left(W_{[s,q]}\in[W_s-\theta,W_s+\theta]|W_s,W_q\right)}_{p_1}\cdot \underbrace{\mathbbm{P}\left(W_{[q,\tau]}\in[W_s-\theta,W_s+\theta]|W_q,W_\tau\right)}_{p_2}.
\end{align}
Relating the decomposition to the statement of the theorem, $p_1$ follows directly from the parametrisation given and the representation in \cite[Thm. 6.1.2]{phd:p13} of the result in \cite[Prop. 3]{mcap:bpr08}. $p_2$ similarly follows from the representation found in \cite[Thm. 5]{b:pjr15}.
\end{theorem}
\begin{corollary} \label{cor:besselacrej}
Letting $p:=\mathbbm{P}\left(W_{[s,\tau]}\in [W_s\!-\!\theta,W_s\!+\!\theta]\right)$, monotonically convergent upper and lower bounding probabilities ($p^{\uparrow}_n$ and $ p^{\downarrow}_n$ respectively) with the property that $\lim_{n\to\infty} p^{\uparrow}_n \to p$ and $\lim_{n\to\infty} p^{\downarrow}_n \to p$ can be found (where $n_0:=\lceil\sqrt{(\tau-q)+4\theta^2}/4\theta\rceil$),
\begin{align}
& p^{\downarrow}_n
 := \left(\dfrac{1 -\sum^{n}_{j=1}\varsigma_{q-s}(j;W_s-W_q,\theta)+\sum^{n-1}_{j=1}\varphi_{q-s}(j;W_s-W_q,\theta)}{1-\exp\left\{-2\theta[m(W_s-W_q)+\theta]/(q-s)\right\}} \right)\nonumber\\
& \quad\quad\quad \cdot \left(1 + \sum^{n_0+n}_{j=1}\psi_{\tau-q}(j;W_q-W_\tau,\theta,m) + \sum^{n_0+n-1}_{j=1}\chi_{\tau-q}(j;W_q-W_\tau,\theta,m)\big]  \right), \label{eq:pdownarrow}
\end{align}
\begin{align}
& p^{\uparrow}_n
 := \left(\dfrac{1 -\sum^{n}_{j=1}\varsigma_{q-s}(j;W_s-W_q,\theta)+\sum^{n}_{j=1}\varphi_{q-s}(j;W_s-W_q,\theta) }{1-\exp\left\{-2\theta[m(W_s-W_q)+\theta]/(q-s)\right\}}\right)\nonumber\\
& \quad\quad\quad \cdot \left(1 + \sum^{n_0+n}_{j=1}\psi_{\tau-q}(j;W_q-W_\tau,\theta,m) + \sum^{n_0+n}_{j=1}\chi_{\tau-q}(j;W_q-W_\tau,\theta,m)\big]  \right). \label{eq:puparrow}
\end{align}
Furthermore we have
\begin{align}
\dfrac{p^\uparrow_n-p^\downarrow_n}{p^\uparrow_{n-1}-p^\downarrow_{n-1}} =: r_n \leq r\in (0,1), \label{eq:geobound}
\end{align}
and so,
\begin{align}
\bar{K} := \sum^\infty_{i=1} |p^\uparrow_i-p^\downarrow_i| = (p^\uparrow_1-p^\downarrow_1) + \sum^\infty_{i=2} \prod_{j=2}^i r_j \leq \sum^\infty_{i=0} r^i = \dfrac{1}{1-r} <\infty. \label{eq:geosum}
\end{align}
\proof
The summations in the left hand brackets of the sequences (\ref{eq:pdownarrow}) and (\ref{eq:puparrow}) follows from \thmref{thm:besselaccprob} and \cite[Prop. 3]{mcap:bpr08}. The summations in the right hand brackets of the sequences (\ref{eq:pdownarrow}) and (\ref{eq:puparrow}), and the necessary condition on $n_0$, follows from \cite[Corollary 5]{b:pjr15}. The validity of the product form of (\ref{eq:pdownarrow}) and (\ref{eq:puparrow}) follows from \cite[Corollary 1]{b:pjr15}. The bound on the ratio of subsequent bound ranges of $p$ in (\ref{eq:geobound}) follows from the exponential decay in $n$ of $\varsigma(n)$, $\varphi(n)$, $\psi(n)$ and $\chi(n)$ of \thmref{thm:besselaccprob}, and as shown in the proof of \cite[Thm. 6.1.1]{phd:p13} and \cite[Corollary 6.1.3]{phd:p13}. (\ref{eq:geosum}) follows directly from (\ref{eq:geobound}).
\endproof
\end{corollary}
\noindent Having established \thmref{thm:besselaccprob} and \corrolref{cor:besselacrej} we can now construct a (retrospective) rejection sampler in which we simulate $W_q$ (as per the law of a Bessel bridge) and, by means of an algorithmic loop in which the bounding sequences of the acceptance probability are evaluated to  sufficient precision, we make the determination of acceptance or rejection. This is summarised in \algref{alg:qWqsim}, further noting that although the embedded loop is of random length, by \corrolref{cor:besselacrej} we know that it halts in finite expected time ($\bar{K}$ can be interpreted as the expected computational cost of the nested loop, noting that $\mathbbm{E}[\text{iterations}] := \sum^\infty_{i=0} i\mathbbm{P}(\text{halt at step i}) =  \sum^\infty_{i=0} \mathbbm{P}(\text{halt at step i or later}) =\bar{K} $).
\begin{algorithm}[h]
	\caption{Simulating $W_q\sim\mathbbm{W}^{W_s,W_\tau}_{s,\tau}|(W_s,W_\tau,\theta)$, given $q\in[s,\tau]$, the end points ($W_s$ and the extremal value $W_\tau$), and constrained such that $\forall u\in[s,\tau], W_u\in[W_s-\theta,W_s+\theta]$.} \label{alg:qWqsim}
	\begin{enumerate}
	\item $\{b^{(i)}_q\}^3_{i=1}$: Simulate $b^{(1)}_q,b^{(2)}_q,b^{(3)}_q \overset{\text{iid}}{\sim} \N\left(0,\dfrac{|\tau-q|\cdot|q-s|}{(\tau-s)^2}\right)$.\label{alg:qWqsim:proposal}
	\item $W_q$: Set $W_{q}:=W_\tau + (-1)^{\mathbbm{1}(W_\tau<W_s)} \sqrt{(\tau-s)\left[\left(\dfrac{\theta(\tau-q)}{(\tau-s)^{3/2}} + b^{(1)}_q\right)^2 + (b^{(2)}_q)^2 +(b^{(3)}_q)^2\right]}$.
	\item $u$: Simulate $u\sim\U[0,1]$ and set $n=1$.
	\item $p^\downarrow_\cdot,p^\uparrow_\cdot$: While $u\notin [p^\downarrow_n,p^\uparrow_n]$, set $n=n+1$.
	\item $p$: If $u\leq p^\downarrow_n$ accept, else reject and return to \stepref{alg:qWqsim:proposal}.
    \item Return $(q,W_q)$.
	\end{enumerate}
\end{algorithm}


\subsection{Simulation of a single trajectory of constrained Brownian motion}\label{ss:fptfinal}

We now have the constituent elements for \secref{s:emcfd}, in which we simulate multivariate Brownian motion at any desired time marginal, with $d$-dimensional hypercubes inscribing intervals of the state space in which the sample path almost surely lies (layers, more formally defined in \cite{b:pjr15}). Recall from \secref{s:emcfd} that the killing times are determined by a random variable whose distribution depends upon the inscribed layers, and so the presentation of \algref{alg:scaleiterate} necessitates a loop in which the determination of whether the stopping time occurs in the interval is required.\\
\\
We require the user-specified vector $\mathbf{\theta}$ in order to determine the default hypercube inscription size. In practice, as with other MCMC methods, we might often apply a preconditioning matrix to the state space before applying the algorithm.\\
\\
Further note that, due to the strong Markov property, it is user preference whether this algorithm is run in its entirety for every required time marginal, or whether it resets layer information whenever any one component breaches its boundary and re-initialises from that time on according to \algstref{alg:scaleiterate}{alg:restart}.
\begin{algorithm}[h]
	\caption{Simulating constrained Brownian motion at a desired time marginal $(t,W_t)$.} \label{alg:scaleiterate}
	\begin{enumerate}
	\item Input $\mathbf{W}_s$ and $\mathbf{\theta}$.
	\item $\mathbf{\tau}$: For $i\in\{1,\ldots{},d\}$, simulate $(\tau^{(i)},W^{(i)}_{\tau})$ as per \algref{alg:devroyefpg}.
    \item $\hat{\tau}$: Set $\hat{\tau}:=\inf_i\{\tau^{(i)}\}$, set $j:=\{i\in\{1,\ldots{},d\}:\tau^{(i)}=\hat{\tau}\}$. \label{algtauhat}
    \item[*] $t$: If required, simulate $t$ as outlined in \secref{s:emcfd}.
    \item $t$: If $t\notin[s,\hat{\tau}]$, 
    \begin{enumerate}
    \item $(\hat{\tau},W^{(\cdot)}_{\hat{\tau}})$: For $i\in\{1,\ldots{},d\}\setminus j$, simulate $(\hat{\tau},W^{(i)}_{\hat{\tau}})$ as per \algref{alg:qWqsim}.
    \item $(\tau^{(j)},W^{(j)}_{\tau})$: Simulate $(\tau^{(j)},W^{(j)}_{\tau})$ as per \algref{alg:devroyefpg}. \label{alg:restart}
    \item $s$: Set $s:=\hat{\tau}$, and return to Step \ref{algtauhat}.
    \end{enumerate}
    \item $(t,W^{(\cdot)}_{t})$: For $i\in\{1,\ldots{},d\}$, simulate $(t,W^{(i)}_{t})$ as per \algref{alg:qWqsim}.
    \item Return $(t,W_t)$.
	\end{enumerate}
\end{algorithm}


\section{Path-space Rejection Sampler (PRS) for $\mu_T$}\label{apx:prs}

A path-space rejection sampler for $\mu_T$ can be constructed by drawing from Brownian motion measure, $\mathbf{X}\sim\mathbbm{W}^{\mathbf{x}}_T$, accepting with probability $P(\mathbf{X})$ given by
\begin{align}
{P}(\mathbf{X}) 
& = \underbrace{\exp\left\{\Phi T -\sum^{n_R}_{i=1} L^{(i)}_{\mathbf{X}}\cdot[(\tau_i\wedge T)\!-\!\tau_{i-1}]\right\}}_{=:P^{(1)}(\mathbf{X}) \in [0,1]} \cdot \prod^{n_R}_{i=1}\Bigg[\underbrace{\exp\left\{-\!\!\int^{\tau_i\wedge T}_{\tau_{i-1}}\!\left(\phi(\mathbf{X}_s)\!-\!L^{(i)}_{\mathbf{X}}\right)\ud s\right\}}_{=:P^{(2,i)}(\mathbf{X})}\Bigg] \label{eq:p1decomp} \\
& = \prod^{n_R}_{i=1} \Bigg[\underbrace{\exp\left\{(\Phi - L^{(i)}_{\mathbf{X}})\cdot[(\tau_i\wedge T)\!-\!\tau_{i-1}]\right\}}_{=:P^{(1,i)}(\mathbf{X}) \in [0,1]} \cdot \underbrace{\exp\left\{-\int^{\tau_i\wedge T}_{\tau_{i-1}}\!\left(\phi(\mathbf{X}_s)\!-\!L^{(i)}_{\mathbf{X}}\right)\ud s\right\}}_{=:P^{(2,i)}(\mathbf{X})}\Bigg]. \label{eq:p1decomp2}
\end{align}
The algorithmic pseudo-code for this approach is thus presented in \algref{alg:altuea}.
\begin{algorithm}[h]
	\caption{Path-space Rejection Sampler (PRS) for $
\mu_T$ Algorithm} \label{alg:altuea}
	\begin{enumerate}
	\item Input: $\mathbf{X}_0$.
	\item $R$: Simulate layer information $R\sim\mathcal{R}$ as per \apxref{s:pathspace layer}. \label{alg:altuea:start}\label{alg:cuea:layer}
	\item $P^{(1)}$: With probability $1- \exp\{\Phi T -\sum^{n_R}_{i=1} L^{(i)}_\mathbf{X}\cdot[(\tau_i\wedge T)\!-\!\tau_{i-1}]\}$ reject and return to \stepref{alg:altuea:start}. \label{alg:cuea:pre}
	\item $n_R$: For $i$ in $1\to n_R$,
	\begin{enumerate}
	\item $\mathbbm{U}^{(i)}_R$: Set $j=0$, $\kappa_i=0$, $\xi^{(i)}_0:=\tau_{i-1}$ and $E^{(i)}_1 \sim \text{Exp}(U^{(i)}_\mathbf{X}-L^{(i)}_\mathbf{X})$. While $\sum_j E^{(i)}_j < [(\tau_i\wedge T)\!-\!\tau_{i-1}]$,
	\begin{enumerate}
	\item $\xi^{(i)}_j$: Set $j=j+1$ and $\xi^{(i)}_j = \xi^{(i)}_{j-1}+E^{(i)}_j$.
	\item $\mathbf{X}_{\xi^{(i)}_j}$: Simulate $\mathbf{X}_{\xi^{(i)}_j}\sim \left.\text{MVN}(\mathbf{X}_{\xi^{(i)}_{j-1}},(\xi^{(i)}_j-\xi^{(i)}_{j-1}))|R^{(i)}_\mathbf{X} \right.$.
	\item $P^{(2,i,j)}$: With probability $1-[U^{(i)}_\mathbf{X}\!-\!\phi\big(\mathbf{X}_{\xi^{(i)}_j}\big)]/[U^{(i)}_\mathbf{X}\!-\!L^{(i)}_\mathbf{X}]$, reject path and return to \stepref{alg:altuea:start}.
	\item $E^{(i)}_{j+1}$: Simulate $E^{(i)}_{j+1} \sim \text{Exp}(U^{(i)}_\mathbf{X}-L^{(i)}_\mathbf{X})$.
	\end{enumerate}
		\item $\mathbf{X}_{\tau_i\wedge T}$: Simulate $\mathbf{X}_{\tau_i\wedge T}\sim \text{MVN}\left.(\mathbf{X}_{\xi^{(i)}_{j}},[(\tau_i\wedge T)-\xi^{(i)}_j])|R^{(i)}_\mathbf{X} \right.$.
	\end{enumerate}
	\end{enumerate}
\end{algorithm}
\\\\ 
Crucially, determination of acceptance is made using only a path \textit{skeleton} (as introduced in \cite{b:pjr15}, a  path \textit{skeleton} is a finite-dimensional realisation of the sample path, including a \textit{layer} constraining the sample path, sufficient to recover the sample path at any other finite collection of time points without error as desired). The PRS for $\mu_T$ outputs the skeleton composed of all intermediate simulations,
\begin{align}
\mathcal{S}_\text{PRS}\left(\mathbf{X}\right) :=\left\{\mathbf{X}_0,\left(\left(\xi^{(i)}_j,\mathbf{X}_{\xi^{(i)}_j} \right)^{\kappa_i}_{j=1},R^{(i)}_\mathbf{X}\right)^{n_R}_{i=1}\right\}, \label{eq:altueaskel}
\end{align}
which is sufficient to simulate any finite-dimensional subset of the remainder of the sample path (denoted by $\mathbf{X}^\text{rem}$) as desired without error (as outlined in \cite[\textsection 3.1]{b:pjr15} and \apxref{apx:fpttheory}),
\begin{align}
\rem{\mathbf{X}}_{(0,T)} & \sim\left.\otimes^{n_R}_{i=1}\left(\otimes^{\kappa_i}_{j=1} \mathbbm{W}^{\mathbf{X}[\xi^{(i)}_{j-1},\xi^{(i)}_j]}_{\xi^{(i)}_{j-1},\xi^{(i)}_j}\right)\mvbar R^{(i)}_\mathbf{X}\right..
\end{align}


\section{Killed Brownian Motion (KBM)}\label{apx:kbm}

In \algref{alg:devroyefpg} we detailed an approach to simulate the killing time and location, ($\bar{\tau},\mathbf{X}_{\bar{\tau}}$), for killed Brownian motion. To avoid unnecessary algorithmic complexity, note that we can recover the pair ($\bar{\tau},\mathbf{X}_{\bar{\tau}}$) by a simple modification of \algref{alg:altuea} in which we set $\forall i\, L^{(i)}_\mathbf{X}:=\Phi$ and return the first rejection time. This is presented in \algref{alg:kbm}. A variant in which $L^{(i)}_\mathbf{X}$ is incorporated would achieve greater efficiency, but is omitted for notational clarity.
\begin{algorithm}[h]
	\caption{Killed Brownian Motion (KBM) Algorithm} \label{alg:kbm}
	\begin{enumerate}
	\item Initialise: Set $i=1$, $j=0$, $\tau_0=0$. Input initial value $\mathbf{X}_0$.
	\item $R$: Simulate layer information $R^{(i)}_{\mathbf{X}}\sim\mathcal{R}$ as per \apxref{s:pathspace layer}, obtaining $\tau_i$, $U^{(i)}_\mathbf{X}$. \label{alg:kbm:layer}
	\item $E$: Simulate $E \sim \text{Exp}(U^{(i)}_\mathbf{X}-\Phi)$. \label{alg:kbm:loop}
	\item $\xi_j$: Set $j=j+1$ and $\xi_j = (\xi_{j-1}+E)\wedge \tau_i$.
	\item $\mathbf{X}_{\xi_j}$: Simulate $\mathbf{X}_{\xi_j}\sim \left.\text{MVN}(\mathbf{X}_{\xi_{j-1}},(\xi_j-\xi_{j-1}))|R^{(i)}_\mathbf{X} \right.$.
	\item $\tau_i$: If $\xi_j=\tau_i$, set $i=i+1$ and return to \stepref{alg:kbm:layer}.
    \item $P$: With probability $[U^{(i)}_\mathbf{X}\!-\!\phi\big(\mathbf{X}_{\xi_j}\big)]/[U^{(i)}_\mathbf{X}-\Phi]$ return to  \stepref{alg:kbm:loop}. \label{alg:omit1} 
    \item $(\bar{\tau},\mathbf{X}_{\bar{\tau}})$: Return $(\bar{\tau},\mathbf{X}_{\bar{\tau}})= ({\xi_j},\mathbf{X}_{\xi_j})$, $i_{\bar{\tau}}=i$, $j_{\bar{\tau}}=j$. \label{alg:omit2}
	\end{enumerate}
\end{algorithm}
\\\\
As in the PRS for $\mu_T$ presented in \apxref{apx:prs}, in KBM (\algref{alg:kbm}) we can recover in the interval $[0,\bar{\tau})$ the remainder of the sample path as desired without error as follows (where for clarity we have suppressed full notation, but can be conducted as described in \apxref{apx:fpttheory}),
\begin{align}
\mathcal{S}_\text{KBM}\left(\mathbf{X}\right) :=\left\{\mathbf{X}_0,(\xi_j,\mathbf{X}_{\xi_j})^{j_{\bar{\tau}}}_{j=1},(R^{(i)}_\mathbf{X})^{i_{\bar{\tau}}}_{i=1}\right\}, & \quad\quad\quad \rem{\mathbf{X}}_{(0,T)} \sim\mathbbm{W} | \mathcal{S}_\text{KBM}. \label{eq:kbmskel}
\end{align}


\section{Rejection Sampling based QSMC Algorithm} \label{apx:rqsmc}

In \secref{s:qsmc} we considered the embedding of IS-KBM of 
\algref{alg:is-kbm} within SMC. A similar embedding for the rejection sampling variant (KBM) of \algref{alg:omit2} is considered here as the probability of the killed Brownian motion trajectory of \algref{alg:kbm} remaining alive becomes arbitrarily small as diffusion time increases. As such, if one wanted to approximate the law of the process conditioned to remain alive until large $T$ it would have prohibitive computational cost.\\
\\
Considering the KBM algorithm presented in \apxref{apx:kbm}, in which we simulate trajectories of killed Brownian motion, the most natural embedding of this within an SMC framework is to assign each particle constant un-normalised weight while alive, and zero weight when killed. Resampling in this framework simply consists of sampling killed particles uniformly at random from the remaining alive particle set. The manner in which we have constructed \algref{alg:kbm} allows us to conduct this resampling in continuous time, and so we avoid the possibility of at any time having an alive particle set of size zero. We term this approach (Continuous Time) Rejection Quasi-Stationary Monte Carlo (R-QSMC), and present it in \algref{alg:ctr-qsmc}. In \algref{alg:ctr-qsmc} we denote by $m(k)$ as a count of the number of killing events of particle trajectory $k$ in the time elapsed until the $m^{\text{th}}$ iteration of the algorithm.
\begin{algorithm}[h]
	\caption{(Continuous Time) Rejection Quasi-Stationary Monte Carlo Algorithm (R-QSMC) Algorithm.} \label{alg:ctr-qsmc}
    \begin{enumerate}
	\item \textbf{Initialisation Step ($m=0$)}
	\begin{enumerate}
	\item Input: Starting value, $\hat{\mathbf{x}}$, number of particles, $N$.
	\item $\mathbf{X}^{(\cdot)}_0$: For $k$ in $1$ to $N$ set $\mathbf{X}^{(1:N)}_{t_0}=\hat{\mathbf{x}}$ and $w^{(1:N)}_{t_0}=1/N$.
	\item $\bar{\tau}^{(\cdot)}_1$: For $k$ in $1$ to $N$, simulate $\!\left.\left(\bar{\tau}^{(k)}_1,\mathbf{X}^{(k)}_{\bar{\tau}_1}\right)\mvbar\left(t_0^{(k)},\mathbf{X}^{(k)}_{t_0}\right)\right.\!$ as per \algref{alg:kbm}. \label{r-scale:sub1}
	\end{enumerate}
	\item \textbf{Iterative Update Steps ($m=m+1$)}
	\begin{enumerate}
	\item $\bar{\bar{\tau}}_m$: Set $\bar{\bar{\tau}}_m:=\inf\{\{\bar{\tau}^{(k)}_{m(k)}\}^N_{k=1}\}$, $\bar{k}:=\{k:\bar{\bar{\tau}}_m=\bar{\tau}^{(k)}_{m(k)}\}$.
	\item $K$: Simulate $K \sim \U\{\{1,\ldots{},n\}\setminus \bar{k}\}$.
	\item $\mathbf{X}^{(\cdot)}_{\bar{\bar{\tau}}_m}$: Simulate $\mathbf{X}^{(\bar{k})}_{\bar{\bar{\tau}}_m}\sim \mathbbm{W}| \mathcal{S}^{(K)}_\text{KBM}$ as given by (\ref{eq:kbmskel}) and as per \algref{alg:qWqsim}.
	\item $\bar{\bar{\tau}}_{m+1}$: Simulate $\!\left.\left(\bar{\tau}^{(\bar{k})}_{m(\bar{k})+1},\mathbf{X}^{(\bar{k})}_{\bar{\tau}_{m(\bar{k})+1}}\right)\mvbar\left(\bar{\bar{\tau}}_m,\mathbf{X}^{(\bar{k})}_{\bar{\bar{\tau}}_m}\right)\right.\!$ as per \algref{alg:kbm}. \label{r-scale:sub2}
	\end{enumerate}	
	\end{enumerate}
\end{algorithm}
\\\\
\noindent Iterating the R-QSMC algorithm beyond some time $t^*$ at which point we believe we have obtained convergence, and halting at time $T>t^*$, we can approximate the law of the killed process by the weighted occupation measures of the trajectories (where $\forall t\, w^{(\cdot)}_{t}=1/N$),
\begin{align}
\pi(\!\ud \mathbf{x}) \approx \hat{\pi}(\!\ud \mathbf{x}) 
& := \dfrac{1}{T-t^*}\int^T_{t^*} \sum^N_{k=1}w^{(k)}_t\cdot \delta_{\mathbf{X}^{(k)}_{t}}(\!\ud \mathbf{x}) \ud t. \label{eq:occupation1}
\end{align}
In some instances the tractable nature of Brownian motion will admit an explicit representation of (\ref{eq:occupation1}). If not, one can simply sample the trajectories exactly at equally spaced points to find an unbiased approximation of (\ref{eq:occupation1}), by means detailed in \apxref{ss:fptinter} and \algref{alg:devroyefpg}. In particular, if we let $t_0:=0< t_1<\ldots{}<t_m:=T$ such that $t_i-t_{i-1}:=T/m$, then we can approximate the law of the killed process as we did in (\ref{eq:occupation2}), where $w^{(1:N)}_{t*:T}=1/N$.


\section{Rejection sampling Scalable Langevin Exact (R-ScaLE) algorithm} \label{apx:r-scale}

In \secref{s:scale} we noted that the survival probability of a proposal Brownian motion sample path was related to the estimator $P(\mathbf{X})$ of \apxref{apx:prs} and in (\ref{s:repest}) where we develop a replacement estimator. The construction of control variates in \secref{s:repest} allows us to construct the replacement estimator such that it has good scalability properties. In a similar fashion to the embedding of this estimator within QSMC (\algref{alg:qsmc}) resulting in ScaLE (\algref{alg:scale}), we can embed this estimator with the rejection sampling variant R-QSMC (\algref{alg:ctr-qsmc}) resulting in the \textit{Rejection Scalable Langevin Exact algorithm (R-ScaLE)} which we present in \algref{alg:r-scale}.\\
\\
Note as presented in \algref{alg:r-scale} we may also be concerned with the absolute growth of $\tilde{\Phi}$ (relative to $\Phi$) as a function of $n$ in order to study its computational complexity. Note however, as remarked upon in \apxref{apx:kbm}, if this growth is not favourable one can modify \algref{alg:kbm} to incorporate the additional path-space bound $\tilde{L}^{(i)}_\mathbf{X}$ for each layer. Details of this modification are omitted for notational clarity.
\begin{algorithm}[h]
	\caption{The R-ScaLE Algorithm (as per \algref{alg:ctr-qsmc} unless stated otherwise).} \label{alg:subscale3} \label{alg:r-scale}
	\begin{enumerate}
    \item[0.] Choose $\hat{\mathbf{x}}$ and compute $\nabla\log\pi(\hat{\mathbf{x}})$, $
    \Delta \log\pi(\hat{\mathbf{x}})$. $\tilde{\Phi}$.
    \item[\ref{r-scale:sub1}.] On calling \algref{alg:kbm}
    \begin{enumerate}
    \item Replace $\Phi$ with $\tilde{\Phi}$.
    \item Replace $U^{(i)}_{\mathbf{X}}$ in \stepref{alg:kbm:layer} with $\tilde{U}^{(i)}_{\mathbf{X}}$.
    \item Replace \stepref{alg:omit1} with: Simulate $I,J\overset{\text{iid}}{\sim} \U\{0, \ldots ,n\}$, and with probability $[\tilde{U}^{(i)}_\mathbf{X}\!-\!\tilde{\phi}\big(\mathbf{X}_{\xi_j}\big)]/[\tilde{U}^{(i)}_\mathbf{X}-\Phi]$ return to  \stepref{alg:kbm:loop}.
    \end{enumerate}
    \item[\ref{r-scale:sub2}] As Step \ref{r-scale:sub1}.
	\end{enumerate}
\end{algorithm}


\section{Discrete Time Sequential Monte Carlo Construction}\label{app:smc}

Considering the discrete time system with state space $\dtss[k] = ({C}(h(k-1),hk],\mathcal{Z}_k)$ at discrete time $k$, with the process denoted $\dtproc[k] = (X_{(h(k-1),hk]},\dtaux[k])$ in which the auxiliary variables $\dtaux[k]$ take values in some space $\mathcal{Z}_k$.\\
\\
The ScaLE Algorithm, with resampling conducted deterministically at times $h,2h,\ldots$ coincides exactly with the mean field particle approximation of a discrete time Feynman-Kac flow, in the sense and notation of \cite{bk:fk}, with transition kernel \[M_k(\dtproc[k-1],d\dtproc[k]) = \mathbb{W}^{X_{h(k-1)}}_{h(k-1),hk}(dX_{(h(k-1),hk]}) Q_k(X_{(h(k-1),hk]},d\dtaux[k])\] and a potential function $G_k(\dtproc[k])$, which is left intentionally unspecified to allow a broad range of variants of the algorithm to be included, the property which it must possess to lead to a valid form of ScaLE Algorithm is specified below. Allowing 
\[
\overline{\mathbb{W}}^{x}_{0,hk}(\dtproc[1:k]) = \mathbb{W}_x^{0,hk}(dX_{0:hk})
\prod_{i=1}^k Q_i(X_{(h(i-1),hi]},d\dtaux[i])
\]
and specifying an extended version of the killed process via
\[
\frac{d\overline{\mathbb{K}}^x_{0,hk}}{d\overline{\mathbb{W}}^x_{0,hk}}
(\dtproc[1:k]) \propto \prod_{i=1}^k G(\dtproc[i]).
\]
The validity of such a ScaLE Algorithm depends upon the following identity holding:
\[
\frac{d\mathbb{K}^x_{0,hk}}{d\mathbb{W}^x_{0,hk}}(X_{0:hk}) \propto
\mathbb{E}_{\mathbb{W}^x_{0,hk}}\left[\left.\prod_{i=1}^k G_i(\dtproc[i])\right|X_{0:hk}\right].
\]
It is convenient to define some simplifying notation. We define the law of a discrete time process (in which is embedded a continuous time process taking values in $C[0,\infty)$):
\[ \overline{\mathbb{W}}^x(d\dtproc[]) = \overline{\mathbb{W}}^x_{0,h}(d\dtproc[1]) \prod_{k=1}^\infty  \overline{\mathbb{W}}^{X_{h(k-1)}}_{h(k-1),hk}(d\dtproc[k])\]
and of a family of processes indexed by $k$, $\overline{\mathbb{K}}^x_k$, again incorporating a continuous time process taking values in $C[0,\infty)$, via:
\[
\frac{d\overline{\mathbb{K}}^x_{k}}{d\overline{\mathbb{W}}_x}(\dtproc[]) \propto
\prod_{i=1}^k G_i(\dtproc[i]).
\]
With a slight abuse of notation, we use the same symbol to refer to the associated finite-dimensional distributions, with the intended distribution being indicated by the argument. We also define the \emph{marginal} laws, $\mathbb{W}^x$ and $\mathbb{K}_k^x$ via:
\begin{align*}
\mathbb{W}^x(dX) =& \overline{\mathbb{W}}^x(dX \times (\otimes_{p=1}^{\infty} \mathcal{Z}_p))\\
\mathbb{K}^x_k(dX) =& \overline{\mathbb{K}}^x_k(dX \times (\otimes_{p=1}^{\infty} \mathcal{Z}_p)).
\end{align*}

\begin{proposition}
Under mild regularity conditions (cf. \cite{bk:fk,as:c04}), for any $\varphi : \mathbb{R}^d \to \mathbb{R}$, any algorithm within the framework described admits a central limit in that:
\begin{align*}
\lim_{N \to \infty} \sqrt{N} \left[\frac{1}{N} \sum_{i=1}^N \varphi(X^i_{hk}) -
  \mathbb{K}_k^x(\varphi(X^i_{hk})) \right] \Rightarrow \sigma_{k,G}(\varphi) Z
\end{align*}
where, $Z$ is a standard normal random variable, $\Rightarrow$ denotes convergence in distribution, and:
{\small
\begin{align*}
\sigma_k^2(\varphi) =& \mathbb{E}_{\overline{\mathbb{W}}} 
\left[
  \left( 
\frac{G_1(\dtproc[1]) \mathbb{E}_{\overline{\mathbb{W}}^x}\left[\prod_{i=2}^k
    G(\dtproc[i])|\dtproc[1] \right]}{\overline{\mathbb{W}}^x(\prod_{i=1}^k G(\dtproc[i]))}
  \right)^2
  \mathbb{E}_{{\mathbb{K}}_{k}^x}
  \left[\left. 
      \left(
        \varphi(X_{hk}) - \mathbb{K}_{k}^x(\varphi(X_{hk})) 
      \right)^2
      \right\vert 
      X_{h}
  \right] 
\right]+\\
& \sum\limits_{p=2}^{k-1}
\mathbb{E}_{\overline{\mathbb{K}}_{p-1}^x} \left[
\left( 
\frac{\overline{\mathbb{W}}^x(\prod_{i=0}^{p-1} G(\dtproc[i]))}{\overline{\mathbb{W}}^x(\prod_{i=0}^k G(\dtproc[i]))} G(\dtproc[p]) \mathbb{E}_{\mathbb{W}^x}\left[\left.\prod_{i={p+1}}^k G(\dtproc[i]) \right| X_{hp}\right]
\right)^2 
\mathbb{E}_{\mathbb{K}_{k}^x}
  \left[\left. \left(\varphi(X_{hk}) - \mathbb{K}_{k}^x(\varphi(X_{k})) \right)^2
  \right| X_{hp}\right] 
\right]+\\
& 
\mathbb{E}_{\overline{\mathbb{K}}_{k-1}^x} \left[
\left( 
\frac{\overline{\mathbb{W}}^x(\prod_{i=0}^{k-1} G(\dtproc[i]))}{\overline{\mathbb{W}}^x(\prod_{i=0}^k G(\dtproc[i]))} G(\dtproc[k])
\right)^2 
\left(\varphi(X_{hk}) - \overline{\mathbb{K}}_{k}^x(\varphi(X_{hk})) \right)^2
\right]
\end{align*}}
\end{proposition}

\begin{proof}[Proof Outline]
It follows by a direct application of the argument underlying the Proposition of \cite{spl:jd08} (which itself follows from simple but lengthy algebraic manipulations from the results of \cite{bk:fk,as:c04}) that for any test function, $\varphi: \mathbb{R}^d \to \mathbb{R}$ satisfying mild regularity conditions (cf. \cite{bk:fk,as:c04}) that 
\begin{align*}
\lim_{N \to \infty} \sqrt{N} \left[\frac{1}{N} \sum_{i=1}^N \varphi(X^i_{hk}) -
  \mathbb{K}_k^x(\varphi(X^i_{hk})) \right] \Rightarrow \sigma_{k,G}(\varphi) Z
\end{align*}
where, $Z$ is a standard normal random variable, $\Rightarrow$ denotes convergence in distribution, and:
\begin{align*}
\sigma_{k,G}^2(\varphi) =& \mathbb{E}_{\overline{\mathbb{W}}} 
\left[
  \left( 
    \frac{d\overline{\mathbb{K}}_{k}^x}{d\overline{\mathbb{W}}^x}  (X_{(0,h]},\dtaux[1]) 
  \right)^2
  \mathbb{E}_{\overline{\mathbb{K}}_{k}^x}
  \left[\left. 
      \left(
        \varphi(X_{hk}) - \overline{\mathbb{K}}_{k}^x(\varphi(X_{hk})) 
      \right)^2
      \right\vert 
      \overline{\mathcal{F}}_{1}
  \right] 
\right]+\\
& \sum\limits_{p=2}^{k-1}
\mathbb{E}_{\overline{\mathbb{K}}_{p-1}} \left[\left( \frac{d\overline{\mathbb{K}}_{k}^x}{d\overline{\mathbb{K}}_{p-1}^x}
  (X_{(0,hp]},\dtaux[1:p]) \right)^2 
\mathbb{E}_{\overline{\mathbb{K}}_{k}^x}
  \left[\left. \left(\varphi(X_{hk}) - \overline{\mathbb{K}}_{k}^x(\varphi(X_{k})) \right)^2
  \right\vert \overline{\mathcal{F}}_{p}\right] 
\right]+\\
& 
\mathbb{E}_{\overline{\mathbb{K}}_{k-1}} \left[\left( \frac{d\overline{\mathbb{K}}_{k}^x}{d\overline{\mathbb{K}}_{k-1}}
  (X_{(0,hk]},\dtaux[1:k]) \right)^2 
\left(\varphi(X_{hk}) - \overline{\mathbb{K}}_{k}^x(\varphi(X_{hk})) \right)^2
\right]
\end{align*} 
with $\{\overline{\mathcal{F}}_p\}_{p\geq0}$ being the natural filtration associated with $\overline{\mathbb{W}}^x$.\\
\\
This can be straightforwardly simplified to:
{\small
\begin{align*}
\sigma_k^2(\varphi) =& \mathbb{E}_{\overline{\mathbb{W}}} 
\left[
  \left( 
\frac{G_1(\dtproc[1]) \mathbb{E}_{\overline{\mathbb{W}}^x}\left[\prod_{i=2}^k
    G(\dtproc[i])|\dtproc[1] \right]}{\overline{\mathbb{W}}^x(\prod_{i=1}^k G(\dtproc[i]))}
  \right)^2
  \mathbb{E}_{{\mathbb{K}}_{k}^x}
  \left[\left. 
      \left(
        \varphi(X_{hk}) - \mathbb{K}_{k}^x(\varphi(X_{hk})) 
      \right)^2
      \right\vert 
      X_{h}
  \right] 
\right]+\\
& \sum\limits_{p=2}^{k-1}
\mathbb{E}_{\overline{\mathbb{K}}_{p-1}^x} \left[
\left( 
\frac{\overline{\mathbb{W}}^x(\prod_{i=0}^{p-1} G(\dtproc[i]))}{\overline{\mathbb{W}}^x(\prod_{i=0}^k G(\dtproc[i]))} G(\dtproc[p]) \mathbb{E}_{\mathbb{W}^x}\left[\left.\prod_{i={p+1}}^k G(\dtproc[i]) \right| X_{hp}\right]
\right)^2 
\mathbb{E}_{\mathbb{K}_{k}^x}
  \left[\left. \left(\varphi(X_{hk}) - \mathbb{K}_{k}^x(\varphi(X_{k})) \right)^2
  \right| X_{hp}\right] 
\right]+\\
& 
\mathbb{E}_{\overline{\mathbb{K}}_{k-1}^x} \left[
\left( 
\frac{\overline{\mathbb{W}}^x(\prod_{i=0}^{k-1} G(\dtproc[i]))}{\overline{\mathbb{W}}^x(\prod_{i=0}^k G(\dtproc[i]))} G(\dtproc[k])
\right)^2 
\left(\varphi(X_{hk}) - \overline{\mathbb{K}}_{k}^x(\varphi(X_{hk})) \right)^2
\right]
\end{align*}}
\end{proof}

\noindent We conclude with the following corollary, showing that the particular combination of sub-sampling scheme and path space sampler fits into this framework and providing its particular asymptotic variance expression.
\begin{corollary}
  Such a CLT is satisfied in particular:
  \begin{enumerate}[(a)]
    \item 
      If no sub-sampling is used and one evaluates the \emph{exact} (intractable) killing rate (as described in \algref{alg:qsmc}).
    \item 
      If sub-sampling is employed within the construct of the layered path-space rejection sampler (as described in \algref{alg:scale}).
      \end{enumerate}
\end{corollary}
\begin{proof}
  Both claims follow directly by the above argument with the appropriate identifications.

  (a) is established by setting:
  \begin{align*}
    \mathcal{Z}_k =& \emptyset
    &
    G_k(\dtproc[k]) =& G(X_{[h(k-1),hk]})\\
    && \propto& \frac{d\mathbb{K}^{X_{h(k-1)}}_{h(k-1),hk}}{d\mathbb{W}^{X_{h(k-1)}}_{h(k-1),hk}}(X_{(h(k-1):hk]})
  \end{align*}

  (b) is established by setting (where we denote by $c$ the number of pairs of data points employed by the subsampling mechanism; $c=1$ for the examples in this paper): 
  \begin{align*}
    \mathcal{Z}_k =& \cup_{m_k=1}^\infty  \otimes_{p=1}^{m_k} R(\tau_{k,p-1},\tau_{k,p})\\
      R(s,t) =& \cup_{\kappa=0}^\infty \{\kappa\} \times (s,t]^\kappa \times \{1,\ldots,n\}^{2c\kappa} \\
        \dtaux[k] =& (r_{k,1},\ldots,r_{k,m_k}) \\
        r_{k,p} =&(
\kappa_{k,p},
        \xi_{k,p,1},\ldots,\xi_{k,p,\kappa_{k,p}},s_{k,j,1,1:2c},\ldots,s_{k,p,\kappa_{k,p},1:2c}) \\
        G_k(\dtproc[k]) =&\exp\left(-\sum_{p=1}^{m_k} L_\theta(X_{\tau_k,p-1})(\tau_{k,p} - \tau_{k,p-1})\right)\\
        &\qquad\cdot \prod_{p=1}^{m_k} \prod_{j=1}^{\kappa_{k,p}}
        \left[
\frac{U_\theta(X_{\tau_{k,p-1}}) - \widetilde{\phi}(X_{\xi_{k,p,j}},s_{k,p,j,1:2c})}{U_\theta(X_{\tau_{k,p-1}}) - L_\theta(X_{\tau_{k,p-1}})} 
\right]
         \\
         Q_k(X_{(h(k-1),hk]},d\dtaux[k]) =& \prod_{p=1}^{m_k} \left[ \textsf{PP}(d\xi_{k,p,1:\kappa_{k,p}}; ( U_\theta(X_{\tau_{k,p-1}}) - L_\theta(X_{\tau_{k,p-1}}) ), [\tau_{k,p-1},\tau_{k,p}])\right.\\
             & \left. \qquad \cdot
\prod_{j=1}^{\kappa_{k,p}} \frac{1}{n^{2c}} \prod_{l=1}^{2c} \delta_{\{1,\ldots,n\}} (ds_{k,j,l}))
             \right] 
  \end{align*}
  where $\textsf{PP}(\cdot;\lambda,[a,b])$ denotes the law of a homogeneous Poisson process of rate $\lambda$ over interval $[a,b]$, $\delta_{\{1,\ldots,n\}}$ denotes the counting measure over the first $n$ natural numbers and a number of variables which correspond to deterministic transformations of the $X$ process have been defined to lighten notation:
  \begin{align*}
    \tau_{k,p} =
    \left\{
    \begin{array}{ll}
      (k-1)h & p = 0\\
      \inf \{ t : | X_t - X_{\tau_{k,p-1}}| \geq \theta\}  & p=1,\ldots,m_k-1\\
      kh & p=m_k
    \end{array}\right.
  \end{align*}
and $m_k$ is the number of distinct layer pairs employed in interval $k$ of the discrete time embedding of the algorithm (i.e. it is the number of first passage times simulated within the continuous time algorithm after time $(k-1)h$ until one of them exceeds $kh$; as detailed in Appendices \ref{ss:fpt} and \ref{ss:fptinter}).
\end{proof}

\section{Estimation of Effective Sample Size}\label{app:ESS}

Assume QSMC (or ScaLE) has been run for an execution (diffusion) time of length $T$, and that the weighted particle set (of size $N$) is to be used at the following auxiliary mesh times $t^*,\ldots{},t_m := T$ (recalling from \secref{s:qsmc} that $t^*\in(t_0,\ldots{},t_m)$ is a user selected quasi-stationary burn-in time) for computation of the Monte Carlo estimators (\ref{eq:occupation2}, \ref{eq:occupation3}).\\
\\
The posterior mean for the parameters at time $t_i\in[t^*,T]$ is simply estimated using the particle set by $\hat{\mathbf{X}}_{t_i}=\sum_{k=1}^N w^{(k)}_{t_i}\cdot \mathbf{X}^{(k)}_{t_i}$. An overall estimate of the posterior mean and variance can be computed as follows:
\begin{align}
 \bar{\mathbf{X}} & =\dfrac{1}{m(T-t^*)/T}\sum_{i=m(T-t^*)/T}^m \hat{\mathbf{X}}_{t_i}, \\
 \hat{\sigma}^2_{\mathbf{X}} & =
\dfrac{1}{m(T-t^*)/T}\sum_{i=m(T-t^*)/T}^m \sum_{k=1}^N w^{(k)}_{t_i}\Big(\mathbf{X}^{(k)}_{t_i}- \bar{\mathbf{X}}\Big)^2,
\end{align}
The marginal ESS for particles at a single time point can be estimated as the ratio of the variance of $\hat{\mathbf{X}}_{t}$ to the estimate of the posterior variance,
\begin{align}
 \mbox{ESS}_{M}=\hat{\sigma}^2_{\mathbf{X}}\left(\frac{1}{m(T-t^*)/T} \sum_{t=m(T-t^*)/T}^m \Big(\hat{\mathbf{X}}_{t_i}-\bar{\mathbf{X}}\Big)^2
\right)^{-1}.
\end{align}
Although in total we have $(m(T-t^*)/T)$ sets of particles (after burn-in), these will be correlated. This is accounted for using the lag-1 auto correlation of $\hat{\mathbf{X}}_{t^*},\ldots,\hat{\mathbf{X}}_{T}$, which we denote $\hat{\rho}$. Our overall estimated ESS is,
\begin{align}
  \mbox{ESS} := (m(T-t^*)/T)\cdot \dfrac{1-\hat{\rho}}{1+\hat{\rho}} \cdot \mbox{ESS}_{M}.
\end{align}

\end{document}